\documentclass[smallextended]{svjour3}
\makeatletter
\def\ps@pprintTitle{%
 \let\@oddhead\@empty
 \let\@evenhead\@empty
 \def\@oddfoot{}%
 \let\@evenfoot\@oddfoot}
\makeatother
\usepackage{graphicx}
\usepackage{amsmath}
\usepackage{amssymb}
\usepackage{graphics}
\usepackage{makeidx}
\usepackage{mathtools}
\mathtoolsset{showonlyrefs=true}

\smartqed % flush right qed marks, e.g. at end of proof
\usepackage{dsfont}
\usepackage{amsmath}
\usepackage{graphicx}
\usepackage{float}
\usepackage{xcolor}
\usepackage{bbm}
\usepackage[toc,title,page]{appendix}
\usepackage{amsfonts}
\usepackage[margin=1in]{geometry}

\graphicspath{{M:/Documents/PhD/Matlab/SIR/QSDfigs/}{C:/Users/hscovert/Documents/LatexDocuments/QSDapproximation/figs/}{C:/Users/hscovert/Documents/LatexDocuments/figs/}{C:/Users/hscovert/Documents/PhD/Matlab/SIR/QSDfigs/}{../../Matlab/SIR/QSDfigs/}{/figs/}{../figs/}}

\begin{document}
\title{Approximating quasi-stationary behaviour in network-based SIS dynamics}
%\author[1,5,6]{Christopher E. Overton\corref{cor1}\fnref{fn1}}
%\ead{c.overton@liverpool.ac.uk}
%\ead{christopher.overton@manchester.ac.uk}
%\author[2]{Robert R. Wilkinson}
%\author[3,1]{Adedapo Loyinmi}
%\author[4]{Joel C. Miller}
%\author[1]{Kieran J. Sharkey}
%
%\address[1]{Department of Mathematics, University of Liverpool, L69 7ZL, United Kingdom}
%\address[2]{Department of Applied Mathematics, Liverpool John Moores University, L3 3AF, United Kingdom}
%\address[3]{Tai Solarin University of Education, P.M.B 2118, Nigeria}
%\address[4]{Department of Mathematics and Statistics, La Trobe University, Bundoora Victoria 3086, Australia}
%\address[5]{Department of Mathematics, University of Manchester, M13 9PY, United Kingdom}
%\address[6]{Clinical Data Science Unit, Manchester University NHS Foundation Trust, M13 9WL, United Kingdom}
%
%\cortext[cor1]{Corresponding author}
%\fntext[fn1]{Present address: Department of Mathematics, University of Manchester, M13 9PY, United Kingdom}
%\date{}

%\titlerunning{Short form of title} % if too long for running head

\author{Christopher E. Overton \and
Robert R. Wilkinson \and
Adedapo Loyinmi \and
Joel C. Miller \and
Kieran J. Sharkey
}

%\authorrunning{Short form of author list} % if too long for running head

\institute{Christopher E. Overton \at
Department of Mathematics, University of Liverpool, UK \\
\emph{Present address:} Department of Mathematics, University of Manchester, UK \\
Clinical Data Science Unit, Manchester University NHS Foundation Trust \\
\email{c.overton@liverpool.ac.uk \\ christopher.overton@manchester.ac.uk} % \\
% \emph{Present address:} of F. Author % if needed
\and
Robert R. Wilinson \at
Department of Applied Mathematics, Liverpool John Moores University, UK
\and
Adedapo Loyinmi \at
Tai Solarin University of Education, Nigeria
\and
Joel C. Miller \at
Department of Mathematics and Statistics, La Trobe University, Australia
\and
Kieran J. Sharkey \at
Department of Mathematics, University of Liverpool, UK
}

\date{Received: date / Accepted: date}

\maketitle
\begin{abstract}
Deterministic approximations to stochastic Susceptible-Infectious-Susceptible models typically predict a stable endemic steady-state when above threshold. This can be hard to relate to the underlying stochastic dynamics, which has no endemic steady-state but can exhibit approximately stable behaviour. Here we relate the approximate models to the stochastic dynamics via the definition of the quasi-stationary distribution (QSD), which captures this approximately stable behaviour. We develop a system of ordinary differential equations that approximate the number of infected individuals in the QSD for arbitrary contact networks and parameter values. When the epidemic level is high, these QSD approximations coincide with the existing approximation methods. However, as we approach the epidemic threshold, the models deviate, with these models following the QSD and the existing methods approaching the all susceptible state. Through consistently approximating the QSD, the proposed methods provide a more robust link to the stochastic models.
\end{abstract}
%\begin{keyword}
%moment-closure \sep graph \sep epidemic model \sep stochastic \sep pair approximation
%\end{keyword}


\section{Introduction}
The Markovian network-based Susceptible-Infectious-Susceptible (SIS) model (also referred to as the contact process~\cite{Harris1974,Liggett1985}) is a stochastic model describing how pathogens spread on a host contact network~\cite{Boccalleti2006,Hadjichrysanthou2015,Kissetal2017,Parshani2010,Pastor2015,Pastor2001,Rock2014}.  In these dynamics, individuals can flip back and forth between two states: susceptible and infected. When an individual is infected, its neighbours in the network (or graph) that are susceptible are directly at risk of becoming infected. Infected individuals eventually return to the susceptible state and are again at risk. If all individuals are susceptible, they remain so for all future time and the pathogen is said to have died out. The all-susceptible state is thus an absorbing state. The model is sometimes considered to be a reasonable mathematical representation for the propagation of sexually transmitted diseases and computer viruses~\cite{Eames2002}.

Approximations to stochastic SIS models, such as mean-field models~\cite{Lajmanovich1976,Pastor2001,vanMieghem2011,vanMieghem2009,Wang2003} and pair-approximation models~\cite{Frasca2016,Hadjichrysanthou2015,Keeling1999,Keeling2005,Kissetal2017,Mata2013,Sharkey2011}, can characterise important features of the stochastic dynamics. One example is the epidemic threshold, below which the pathogen quickly goes extinct, and above which large outbreaks can occur. However, above threshold, these approximate models reach a stable endemic steady-state solution~\cite{Parshani2010} which is not observed in the stochastic dynamics. The stable endemic steady-state which emerges means it is not always clear how to relate these results back to the underlying stochastic process, since the only stable solution to the stochastic model is the disease-free state. 

Sufficiently above threshold, the stochastic system may exhibit apparently stationary behaviour, since the probability of extinction over any finite time period can be made very small. Indeed, this apparently stationary behaviour is often observed, with extinction just a theoretical certainty which almost never occurs over reasonable timescales. The quasi-stationary distribution (QSD) is commonly used to define, quantify and understand the long-term behaviour of finite Markov chains with absorbing states. Examples include: modelling the spread of a computer virus across a network with cure and reinfection~\cite{Kephart1993,Murray1988,Pastor2001,Wierman2004}, chemical reactions in which materials or catalysts can be exhausted~\cite{Dambrine1981note1,Dambrine1981note2,Oppenheim1977,Parsons1987,Pollett1988}, and wildlife management models~\cite{Holling1973,Klein1968,Mech1966,Pakes1987,Pollett1987,Pollett1995,Scheffer1951}. Within Markovian SIS dynamics, various statistics have been derived using the concept of the QSD~\cite{AnderssonBritton2000,Artalejoetal2010,Artalejoetal2013,Hagenaarsetal2004}. This includes use by Wilkinson and Sharkey~\cite{Wilkinson2013} to derive a measure of the invasion probability, by Ferreira and colleagues~\cite{Ferreiraetal2012} to approximate the epidemic threshold, and by N{\aa}sell~\cite{Nasell1999time} to account for the influence of epidemic and demographic forces on the time to extinction.

The calculation of the QSD can require a large number of stochastic simulations, and therefore it is necessary to derive approximation methods. Thus far, approximations have mainly focused on well-mixed populations. Kriscio and Lefevre~\cite{Kryscio2004} used a conditional birth-and-death process to approximate the QSD of Markovian SIS epidemic dynamics, which has since been extended by N{\aa}sell~\cite{Nasell1996,Nasell1999}. Allen and Burgin~\cite{Allen2000} used a system of ordinary differential equations that approximate the expected prevalence in the QSD when the epidemic severity is high. Dickman and Vidigal~\cite{Dickman2002} developed a pair approximation for the QSD on circles, which the model derived in this paper yields as a special case.

In arbitrary network-structured populations, van Mieghem~\cite{vanMieghem2011} has shown that the endemic steady-state of the ``$N$-intertwined" individual-based SIS model, which is akin to the network-based mean-field approximation, leads to a ``meta-stable state", which is consistent with the quasi-stationary distribution, when sufficiently above the epidemic threshold. This behaviour has also been observed for pair-based SIS approximations~\cite{Hadjichrysanthou2015,Kissetal2017}. These approximations to the stochastic dynamics are typically obtained by making statistical independence assumptions. 

%{\color{red}These network-structured individual-based and pair-based method only reach this endemic steady-state when the parameters satisfy certain criteria; i.e. the system is above the epidemic threshold. The aim of this paper is to extend these methods to approximate the expected number of infected individuals in the QSD for all parameter values. By doing so, we develop a system of differential equations that are directly linked to the underlying stochastic process for all parameter values. This addresses the discrepancy with existing models, in which it is not clear what the approximation is capturing. Additionally, through comparing the QSD model to the existing methods, we provide further insight into the nature of these approximations. }

Here our objective is to clarify the link between stochastic SIS dynamics and the approximate models by relating them via the QSD. Well above threshold, the stochastic model exhibits stationary-like behaviour and the conditioning of the QSD has minimal impact over short timescales, yielding a meaningful connection between the stochastic model and its approximations. Closer to the threshold and below threshold, the mapping according to the QSD becomes more important because the unconditioned approximate model with its steady state no longer approximates the stochastic process and its absorbing state. This leads to greater numerical correspondence between the models in this regime.

The paper is structured as follows. In Section~\ref{sec:SIS}, we define the Markovian network-based SIS modelling framework and the master equation that describes the expected behaviour, followed by defining the QSD in Section~\ref{sec:QSD}. Section~\ref{sec:approx} develops approximation methods that capture aspects of the QSD in a numerically feasible way and we prove the existence of endemic equilibrium solutions for the node-level pair-based SIS approximation (often referred to as the pair-quenched-mean-field approximation). We then analyse the performance of the proposed methods on different contact networks in Section~\ref{sec:res}. 

\section{Markovian \textbf{SIS} dynamics on a contact network}
\label{sec:SIS}
We consider a finite set $\mathcal{V}$ of individuals, who are labelled via an arbitrary bijection to $ \{1,2, \ldots , |\mathcal{V}|\}$. Let $N=|\mathcal{V}| < \infty$. Individuals can be in one of two states: susceptible, denoted by $S$, or infected, denoted by $I$. An individual $j \in \mathcal{V}$, while infected, makes infectious contacts to an individual $i \in \mathcal{V} \setminus \{ j \}$ according to a Poisson process with rate $T_{ij}\geq 0$. If a susceptible individual $k \in \mathcal{V}$ receives an infectious contact, it immediately becomes infected for an exponentially distributed time period with mean $1/\gamma_k$, after which it immediately becomes susceptible again. We define the neighbourhood of an individual $j$, denoted $\mathcal{N}_j$, as the set of individuals that can either infect or be infected by $j$; i.e. $i \in \mathcal{N}_j$ if $T_{ij}>0$ or $T_{ji}>0$. We assume that the transmission matrix $T$ is strongly connected; i.e. every individual is at risk of future infection if at least one individual is infected. The matrix $T$ can either represent a directed or undirected contact network.

This model is described by a continuous-time Markov chain $\{ \Sigma(t): t \ge 0 \}$ with finite state space $\{S,I\}^N$, parametrised by an irreducible square matrix $T$ with non-negative entries and a vector $\gamma$ with positive entries. Let $\sigma_\alpha \in \{S,I\}^N$ denote a state of the population.
%, which we label via an arbitrary bijection to $ \{\sigma_1,\sigma_2, \ldots , \sigma_{2^N}\}$, except that 
We assume throughout that state $\sigma_1$ corresponds to the all susceptible state. Let $\Sigma_i(t)$ denote the status of individual $i$ at time $t$, and for a given state $\sigma_\alpha$, let $\sigma_{\alpha i}$ denote the status of individual $i$ in that state. 

From a given state $\sigma_\alpha$, the process can transition to a new state in which one individual has changed state from $S$ to $I$ or from $I$ to $S$. If the status of individual $i$ is changing, we denote the new state by $\sigma_\alpha^{i\to X}$, where $X \in \{S,I\}$ is the new status of $i$. The transition rates for the Markov chain are given in Table \ref{table1}, where $\delta$ is the Kronecker delta.
\begin{table}[h]
\begin{center}
\caption{Transitions for the Markovian network-based SIS model}
\smallskip
\smallskip
\begin{tabular}{cccc}
\hline\noalign{\smallskip}
from & to & at rate\\
\noalign{\smallskip}\hline\noalign{\smallskip}
$\sigma_\alpha:\sigma_{\alpha i}=S$ & $\sigma_\alpha^{i \to I}$ & $ \sum_{j \in \mathcal{V}} T_{ij} \delta_{I\sigma_{\alpha j}} $ \\
$\sigma_\alpha:\sigma_{\alpha i}=I$ & $\sigma_\alpha^{i \to S}$ & $\gamma_i $ \\ 
\noalign{\smallskip}\hline
\end{tabular}
\label{table1}
\end{center}
\end{table}

\noindent The time evolution of the Markov chain is captured by the master equation
\begin{equation}
\frac{\mathrm{d} P(t)}{\mathrm{d} t}= QP(t) \label{eqn:ME},
\end{equation}
where $P_{\alpha}(t)=P(\Sigma(t)=\sigma_\alpha)$ is the probability that the system is in state $\sigma_\alpha$ at time $t \ge 0$, and $Q$ is a matrix of transition rates (obtained from Table \ref{table1}). In particular, $P_1(t)$ denotes the probability that all individuals are susceptible at time $t$. Although this can be solved to determine the future behaviour, in many cases this is infeasible since the matrix $Q$ grows rapidly with $N$.

\section{The quasi-stationary distribution}
\label{sec:QSD}
Let us construct a vector $\rho(t)$, such that its components $\rho_\alpha(t)$, indexed by $\alpha$, represent the conditional probability that the system is in state $\sigma_\alpha$ at time $t$ given that at least one individual is infected; i.e. $\rho_\alpha(t)=P(\Sigma(t)=\sigma_\alpha|\Sigma(t)\neq\sigma_1)$, where $\sigma_1$ is the disease-free state. We have
\begin{equation}
\rho_\alpha(t)=\frac{P_{\alpha}(t)}{1-P_{1}(t)},
\label{eqn:CondDef}
\end{equation}
for $\alpha \neq 1$. For $\alpha=1$, we set $\rho_1(t)=0$ for all $t$. Here we have assumed that $P_{1}(t) \neq 1$ for all $t \ge 0$, which is satisfied whenever $P_{1}(0) \neq 1$, though as $t \to \infty$ the limit tends to 1, which can eventually make it numerically unstable to calculate the conditional probability this way. Using Equation~\eqref{eqn:CondDef} and the master equation \eqref{eqn:ME}, the time derivative of $\rho_{\alpha}(t)$ is given by
\begin{equation}
\frac{\mathrm{d} \rho_\alpha}{\mathrm{d} t}=
\begin{cases}
0 &\text{if } \alpha=1\\
\frac{(QP)_\alpha}{1-P_1}+\frac{P_\alpha(QP)_{1}}{(1-P_1)^2} &\text{if } \alpha =2,3, \ldots , 2^N,
\end{cases}
\label{eqn:CondME}
\end{equation}
%Equivalently,
%\begin{equation}
%\frac{\mbox{d} \rho_\sigma}{\mbox{d} t}=
%\begin{cases}
%0 &\text{if } \sigma=1\\
%(Q\rho)_\sigma+\rho_\sigma(Q\rho)_{1} &\text{if } \sigma =2,3, \ldots , 2^N.
%\end{cases}
%\label{eqn:CondME2}
%\end{equation}
where we suppress the explicit time dependence of $P$ and $\rho$ in favour of compactness. The right-hand side can be expressed in terms of $\rho$ by using Equation~\ref{eqn:CondDef}. However, we opt to keep this in terms of $P$ since this form is used when developing the approximate models.

The state space for the Markov chain is finite and consists exhaustively of one absorbing state and a communicating class of transient states. The non-absorbing states form a communicating class of transient states because the contact network is strongly connected and the vector ${\gamma}$ of recovery rates is positive. Thus, there exists a unique quasi-stationary distribution (QSD)~\cite{Darroch1967}, independent of initial conditions, which is equivalent to the limiting conditional distribution. This QSD, denoted by $\rho^*$, is a stationary distribution of the conditional probability and an equilibrium of Equation~\eqref{eqn:CondME}. Since $\rho^*$ is unique, if we find some distribution $P^*$ over all $2^N$ system states which satisfies
\begin{equation}
\frac{(QP^*)_\alpha}{1-P^*_1}+\frac{P^*_\alpha(QP^*)_{1}}{(1-P^*_1)^2}=0 \qquad \alpha =2,3, \ldots , 2^N,
\label{eqn:QSD}
\end{equation}
then $\rho^*$ is given by
\begin{equation}
\rho^*_\alpha=
\begin{cases}
0 &\text{if } \alpha=1\\
\frac{P_\alpha^*}{1-P^*_1}, &\text{if } \alpha =2,3, \ldots , 2^N.
\end{cases}
\label{eqn:QSDdef}
\end{equation}
Here, for convenience, we define the QSD such that it assigns probability zero to the absorbing state, as opposed to leaving it undefined. Finding the QSD directly is in many cases infeasible since the size of the state space grows geometrically with the population size.

To go from the system-level master equation to node-level equations, we sum Equation~\eqref{eqn:CondME} over all states in which individual $i \in \mathcal{V}$ is infected. Through this (see Appendix~\ref{App:nodelevel}), we arrive at an expression for the rate of change of the probability that $i$ is infected conditioned on non-extinction
\begin{equation}
\frac{\mathrm{d}}{\mathrm{d}t}\left(\rho(\Sigma_i(t)=I)\right)= \frac{\sum_{j}{T_{ij}} \langle S_iI_j \rangle-\gamma_i \langle I_i \rangle}{1- P_1} + \frac{ \langle I_i \rangle }{(1-P_1)^2}\sum\limits_{j}\gamma_j\langle I_j S \rangle,
\label{eqn:IndivCond}
\end{equation}
where $\langle S_iI_j\rangle$ is shorthand for $P(\Sigma_i(t)=S,\Sigma_j(t)=I)$, $\langle I_i \rangle$ is shorthand for $P(\Sigma_i(t)=I)$ and $\langle I_j S \rangle$ is shorthand for the probability that node $j$ is infected and all other nodes are susceptible. As above, $P_1$ is the probability that all nodes are susceptible. Finding a steady-state solution such that

\begin{equation}
0 = \frac{\sum_{j}{T_{ij}} \langle S_iI_j \rangle-\gamma_i \langle I_i \rangle}{1- P_1} + \frac{ \langle I_i \rangle }{(1-P_1)^2}\sum\limits_{j}\gamma_j\langle I_j S \rangle,
\label{eqn:test}
\end{equation}
the probability that node $i$ is infected in the QSD can be calculated as
\begin{equation}
\langle I_i\rangle^{\rm QSD} = \frac{\langle I^*_i \rangle}{(1-P^*_1)}.
\label{eqn:qsd_adj}
\end{equation}

To find an exact solution to Equation~\eqref{eqn:test} requires constructing a hierarchy describing how different states, ranging from pairs up to full system size, change in time, which is computationally no more efficient than solving Equation~\eqref{eqn:QSDdef} directly. However, in this form we can employ moment-closure techniques to approximate these higher order terms. Such approaches are commonly used for approximating the standard probability distribution for epidemic models~\cite{Frasca2016,Hadjichrysanthou2015,Keeling1999,Keeling2005,Kissetal2017,Mata2013,Sharkey2011}. One approach is to assume statistical independence at the level of indidividuals in Equation~\eqref{eqn:IndivCond}. Alternatively, we can construct exact equations describing how the pair probability $\rho(\Sigma_i(t)=S,\Sigma_j(t)=I)$ changes in time, which we can approximate by assuming statistical independence at the level of pairs. 

Although on the left-hand side of Equation~\eqref{eqn:IndivCond} we define the conditional distribution, we retain the standard distributions on the right-hand side. It is possible to express the right-hand side in terms of conditional probabilities. However, through keeping the standard distributions, the approximations obtained in the later sections were found to be more reliable (not shown). By finding approximations that would cause the right-hand side to be zero, we can then transform these into approximations to the conditional distribution by using Equation~\eqref{eqn:qsd_adj}, where both $\langle I_i^*\rangle$ and $P_1^*$ will also need be approximated.

\section{Approximating the QSD}
\label{sec:approx}
In this section, we use moment closure methods to approximate the solution to Equation~\eqref{eqn:IndivCond}. The first approach is to assume that the states of neighbouring individuals are statistically independent, resulting in a relatively simple model that scales computationally with the number of nodes in the network. We then remove this assumption, and instead assume statistical independence at the level of pairs. This results in a more complex model that scales computationally with the number of pairs of nodes, but should capture the correlations between neighbouring nodes.

\subsection{Individual-based approach}
\label{sec:Indiv}
%\noindent Approximating Equation~\eqref{eqn:IndivCond} by assuming that the states of individuals are independent gives
%\begin{align}
%\frac{\mathrm{d}}{\mathrm{d}t}\left(\rho(\Sigma_i(t)=I)\right)\approx \frac{\sum_{j}{T_{ij}} \langle S_i \rangle \langle I_j \rangle-\gamma_i \langle I_i \rangle}{1- \prod_{k} \langle S_k \rangle} + \frac{ \langle I_i \rangle }{(1-\prod_{k} \langle S_k \rangle)^2}\sum\limits_{j}\gamma_j\langle I_j \rangle \prod_{k \neq j} \langle S_k \rangle. 
%\label{eqn:Rob2}
%\end{align}
\noindent Approximating Equation~\eqref{eqn:test} by assuming that the states of individuals are independent gives
\begin{align}
0\approx \frac{\sum_{j}{T_{ij}} \langle S_i \rangle \langle I_j \rangle-\gamma_i \langle I_i \rangle}{1- \prod_{k} \langle S_k \rangle} + \frac{ \langle I_i \rangle }{(1-\prod_{k} \langle S_k \rangle)^2}\sum\limits_{j}\gamma_j\langle I_j \rangle \prod_{k \neq j} \langle S_k \rangle. 
\label{eqn:Rob2}
\end{align}
%To find the approximation to the probability that node $i$ is infected in the QSD ($\langle I_i\rangle^{\rm QSD}$) under this independence assumption, we first need to find a steady state of Equation~\eqref{eqn:Rob2}, which is given by vectors $\langle X \rangle^*$ and $\langle Y \rangle^*$ satisfying,
To find the approximation to the probability that node $i$ is infected in the QSD ($\langle I_i\rangle^{\rm QSD}$) under this independence assumption, we need to find vectors $\langle X \rangle^*$ and $\langle Y \rangle^*$ satisfying,
\begin{equation}
0= \frac{\sum_{j}{T_{ij}} \langle X_i \rangle^* \langle Y_j \rangle^*-\gamma_i \langle Y_i \rangle^*}{1- \prod_{k} \langle X_k \rangle^*} + \frac{ \langle Y_i \rangle^* }{(1-\prod_{k} \langle X_k \rangle^*)^2}\sum\limits_{j}\gamma_j\langle Y_j \rangle^* \prod_{k \neq j} \langle X_k \rangle^*,
\label{eqn:Rob3}
\end{equation}
for all $i$. In the exact case, we need to scale the steady-state by the density remaining in the transient states (Equation~\eqref{eqn:qsd_adj}) to obtain the QSD probability. Following a similar procedure, from $\langle X \rangle^*$ and $\langle Y \rangle^*$, the probability that $i$ is infected in the QSD is approximated by computing
\begin{equation}\langle I_i \rangle_{ \rm approx}^{\rm QSD}=\frac{\langle Y_i \rangle^*}{1-\prod_{k} \langle X_k \rangle^*}.\label{eqn:adjustment}\end{equation}
However, we are only interested in solutions of Equation~\eqref{eqn:Rob3} that are feasible; i.e. $\langle Y_i \rangle^* \in [0,1], \langle X_i \rangle^*=1-\langle Y_i \rangle^*$ for all $i$. To obtain such a solution, define
\begin{equation} 
\frac{\mathrm{d}\langle{Y_i}\rangle}{\mathrm{d}t}= \sum_{j}{T_{ij}}\langle X_i \rangle \langle Y_j \rangle-\gamma_i \langle Y_i\rangle + \frac{\langle Y_i \rangle \sum\limits_{j}\gamma_j \langle Y_j \rangle \prod\limits_{k \neq j} \langle X_k \rangle}{ 1-\prod\limits_{k} \langle X_k \rangle }  \quad , \quad \langle X_i \rangle=1- \langle Y_i \rangle.
\label{eqn:Rob4}
\end{equation}
Equation~\eqref{eqn:Rob4} is positively invariant in $[0,1]^N$ (see Appendix~\ref{App:invariant}). This gives a system of $N$ coupled equations, which can be numerically integrated until a steady state is reached. Alternatively, other fixed point analysis approaches can be applied. Starting from initial conditions satisfying $\langle X_i \rangle \in [0,1]$ and $\langle Y_i \rangle \in [0,1]$ at $t=0$ for all $i$, this process will give a feasible solution to Equation~\eqref{eqn:Rob3}. From Equation~\eqref{eqn:adjustment}, we can approximate the expected number infected in the QSD as 
\[[I]_{\rm approx}^{\rm QSD}= \sum_{i}\langle I_i \rangle_{\rm approx}^{\rm QSD}.\]
We refer to this as the \textit{node-level individual-based model}.
%Note that $\langle Y_i \rangle=0$ for all $i$ is a steady-state of this system; however, since we are studying the QSD, only non-zero solutions are interesting.

\begin{theorem}
For a (strongly connected) K-regular graph, with homogeneous transmission and recovery rates, the node-level individual-based model yields a feasible approximation of the expected prevalence in the QSD such that $\langle I_i \rangle_{ \rm approx}^{\rm QSD} \in (0,1)$, this being the same for all $i \in \mathcal{V}$, and $[I]_{\rm approx}^{\rm QSD} \in (1,N)$. On any strongly connected network, provided a solution exists such that $\langle I_i \rangle_{ \rm approx}^{\rm QSD} \in (0,1)$ (which is found to hold numerically in all instances investigated), then $[I]_{\rm approx}^{\rm QSD} \in (1,N)$. Therefore, the number of infected individuals in the QSD is lower bounded by 1, a property which is shared by the true QSD. 
\end{theorem}
\begin{proof}
Appendix~\ref{App:existence} 
\end{proof}

As a further approximation to the expected number of infected individuals in the QSD, we can treat all individuals of a given degree equally by assuming
\begin{equation}
\langle I_i \rangle \approx \frac{[I_{k_i}]}{|C_{k_i}|}, \nonumber
\end{equation}
where $k_i$ is the degree of node $i$, $[I_{k_i}]$ is the expected number of infected individuals with degree $k_i$, and $|C_k|$ is the number of degree $k$ nodes. To make this approximation, we must also assume that the contact rate and recovery rate only depend on the degree of the individuals, i.e. $T_{ij}=T_{k_i k_j}$ (whenever $T_{ij} >0$)  and $\gamma_i= \gamma_{k_i}$. After summing over all $i$ of a given degree, assuming statistical independence at the level of individuals, and setting the left hand side to zero, Equation~\eqref{eqn:Rob4} becomes a system of $M$ equations in as many variables, where $M$ is the number of unique node degrees in the network (see Appendix~\ref{App:IndivPop}). We refer to the resulting model as the \textit{population-level individual-based model}. In the special case of a circle network, this population-level model yields a model developed by~\cite{Dickman2002}.

%We have not proven uniqueness, but this holds numerically for all parameters tested. This existence and uniqueness also holds numerically for complex networks. Since there is always a solution in $(0,1)$, the solution only approaches $0$ in the limit $\tau \to 0$ (or $\gamma \to \infty$).
%\begin{theorem}
%As $\tau \to 0$ (or $\gamma \to \infty$), the QSD approximation, $\sum_i \langle Y_i \rangle/( 1-\prod_{k} \langle X_k \rangle)$, approaches 1.
%\end{theorem}
%\begin{proof}
%Appendix~\ref{App:bound}
%\end{proof}
%This is consistent with the true QSD, for which the expected number infected is bounded below by 1, and shows that this model has the required properties. It has been observed that the steady state of the standard individual-based model (Equation~\eqref{eqn:StandardIndiv}) captures the dynamics of the meta-stable state when sufficiently above the epidemic threshold~\cite{vanMieghem2011}. This meta-stable state corresponds with the QSD, and hence in this region the standard model captures the QSD. This can be seen by comparing the standard model to the QSD model. In Equation~(\eqref{eqn:Rob4}), since the probability of each node being susceptible decreases as the transmission rate increases, if the population is sufficiently large the product over susceptible nodes approaches zero. In this case, Equation~(\eqref{eqn:Rob4}) tends towards the standard individual-based model. 

%
\subsection{Pair-based approach}
\label{sec:Pair}
\noindent Assuming independence at the level of individuals may be unrealistic, since infection spreads through contact between individuals. Here, we keep Equation~\eqref{eqn:IndivCond} without approximation, and also sum Equation~\eqref{eqn:CondME} over all states in which individual $i \in \mathcal{V}$ is susceptible and individual $j \in \mathcal{N}_i$ is infected, so that we arrive at the equation for the rate of change of the probability that $i$ is susceptible and $j$ is infected conditioned on non-extinction (following a similar derivation to Equation~\eqref{eqn:IndivCond}):
\begin{equation}
\frac{\mathrm{d}}{\mathrm{d}t}\left(\rho(\Sigma_i(t)=S,\Sigma_j(t)=I)\right)= \frac{\dot{\langle S_i I_j \rangle}}{1- P_1} + \frac{ \langle S_i I_j \rangle }{(1-P_1)^2}\sum\limits_{j}\gamma_j\langle I_j S \rangle \qquad (i \in \mathcal{V}, j \in \mathcal{N}_i),
\label{eqn:IndivCond3}
\end{equation}
where $\dot{\langle S_i I_j \rangle}$ is the rate of change in the probability that $i$ is susceptible and $j$ is infected under the standard distribution. The rate $\dot{\langle S_i I_j \rangle}$ depends on triple-probabilities (see Appendix~\ref{sec:Standard}), which can be approximated in terms of individual-probabilities and pair-probabilities using

\begin{equation}\langle A_i B_j C_k \rangle= {\langle A_iC_k|B_j \rangle}{\langle B_j \rangle} \approx  \frac{ \langle A_i B_j \rangle \langle B_j C_k \rangle}{\langle B_j \rangle},
	\label{pairapprox}
	\end{equation}
which assumes that the states of nodes $i$ and $k$ are independent when given the state of node $j$.
	Guided by this approximation for triple-probabilities, and following~\cite{Frasca2016,Sharkey2015}, we then approximate
	\begin{equation}
	\langle S \rangle \approx \frac{ \prod_{i,j \in \mathcal{N}_i: j < i} \langle S_i S_j \rangle }{\prod_i \langle S_i \rangle^{n_i-1}} \quad , \quad
	\langle I_j S \rangle \approx \frac{\prod\limits_{x \in \mathcal{N}_j}\langle I_jS_x\rangle\prod\limits_{y \neq j}\prod\limits_{x \in \mathcal{N}_y:x<y,x\neq j}\langle S_yS_x \rangle }{\prod\limits_{x \neq j}\langle S_x\rangle^{k_x-1}\langle Y_j \rangle^{k_j-1}}.
% \langle S \rangle \left( \frac{\langle S_i \rangle}{\langle I_i \rangle} \right)^{n_i-1}  \prod_{j \in \mathcal{N}_i}
%	\frac{\langle I_i S_j \rangle}{\langle S_i S_j \rangle} \qquad (i \in \mathcal{V})
	\label{pairapprox2}
	\end{equation}
	Setting the left hand sides of equations~\eqref{eqn:IndivCond} and \eqref{eqn:IndivCond3} to zero, applying the above approximations, and imposing
	\begin{align*}
	\langle S_i \rangle&= 1-\langle I_i \rangle, \\
	\langle S_i S_j \rangle&=\langle S_i \rangle-\langle S_i I_j \rangle,\\
	\langle I_i I_j \rangle&=\langle I_i \rangle-\langle I_i S_j \rangle,
	\end{align*}
	yields a system of $\displaystyle N+ \sum_i k_i$ equations in as many variables (see Appendix~\ref{App:pair}). We refer to this as the \textit{node-level pair-based model}.

Again, as a further approximation we can treat all individuals of a given degree equally, and all pairs of neighbours of given degrees equally, by assuming
\begin{equation} 
\langle I_i \rangle\approx  \frac{[I_{k_i}]}{|C_{k_i}|} \quad , \quad \langle S_i I_j \rangle \approx \frac{[S_{k_i}I_{k_j}]}{|C_{k_i k_j}|}  \quad , \quad T_{ij}=T_{k_i k_j} \quad, \quad \gamma_i= \gamma_{k_i} \qquad (i \in \mathcal{V} , j \in \mathcal{N}_i), \nonumber
\end{equation} 
where $|C_{k,l}|$ is the number of pairs between a degree $k$ node and a degree $l$ node, and $[S_{k_i}I_{k_j}]$ is the expected number of pairs involves a susceptible degree $k_i$ node and an infected degree $k_j$ node. After applying approximations \eqref{pairapprox} and \eqref{pairapprox2}, summing over all $i$ of a given degree and over all pairs $i$ and $j$ of given degrees, and setting the left hand sides to zero, equations~\eqref{eqn:IndivCond} and \eqref{eqn:IndivCond3} become a system of $M+M^2$ equations (see Appendix~\ref{App:population}). We refer to this as the \textit{population-level pair-based model}.

\section{Numerical results}
\label{sec:res}
Here we determine how the methods developed in this paper perform when used to approximate the expected number of infected individuals in the QSD for various networks and parameter values.

We assume that: the transmission rate for any pair of connected individuals is equal (taking $T_{ij}=\tau$ whenever $T_{ij}>0$ and zero otherwise), the network is undirected, and infected individuals recover at the same rate; i.e. $\gamma_i=\gamma$ for all $i \in \mathcal{V}$. In the case of an evenly-mixed population, represented by a complete network, the epidemic threshold of the standard individual-based model is given by $(N-1) \times \tau/\gamma =1$. We therefore choose to plot the expected number of infected individuals against $\tau\times\bar{d}/\gamma$, where $\bar{d}$ is the average degree of the graph, to ensure that all networks are tested over a comparable range of epidemic severity. This is a rough approximation for epidemic severity, since in reality it depends on the degree distribution and correlations rather than just the average degree~\cite{Keeling1999}. We assume $\gamma=1$ throughout, so that the ratio can be changed by changing $\tau$. 

The standard individual-based models (Appendix~\ref{sec:Standard}) have been proven to possess a non-zero steady-state solution in the region of parameter space where the epidemic severity is large~\cite{Kissetal2017,vanMieghem2011}. For the standard node-level pair-based model (Appendix~\ref{sec:Standard}), we prove existence of a non-zero steady-state solution in Appendix~\ref{App:standard}. When the transmission rate is sufficiently large, we observe that the models proposed in this paper converge to the standard models, so these endemic steady-states approximate the expected number of infected individuals in the QSD. To demonstrate this, the dynamics for these standard models are compared to the QSD approximation methods (Section~\ref{sec:approx}). We are particularly interested in how our methods perform for low values of $\tau\times\bar{d}/\gamma$, where the standard models will not capture the QSD.

As a baseline for comparison, we simulate the stochastic SIS model using the Gillespie algorithm. To calculate the expected prevalence in the QSD, we average over all simulations that have not gone extinct. 100,000 simulations are run until $t=300$, since by this point all cases reached a steady-state. We compare the expected number of infected individuals in this solution with the steady-state of the QSD approximation methods and the standard models, solved using the Runge-Kutta method. For both the stochastic simulations and the approximation methods, the population is initiated with every node infected. This is to improve accuracy of the stochastic simulations, since a higher proportion will attain the QSD. 

\subsection{Impact of network structure}
To test the methods, consider three networks: the complete network, the (NxN) square-lattice (with fixed boundaries), and Zachary's karate club~\cite{Zachary1977}. The complete network represents a well-mixed population, in which all individuals are connected to each other. The square-lattice is a commonly used network when adding structure to population dynamics. We consider the variant with fixed boundaries, so the interior nodes have four neighbours, edge nodes have three neighbours and corner nodes have two neighbours. Although there is a lot of symmetry across the network, the regular structure with multiple loops can prove challenging for moment closure approximation methods. Zachary's karate club is an example of a real world network, formed from interactions between members of a karate club.

Figure~\ref{fig:Indiv}(a) compares the node-level individual-based model (Equation~\eqref{eqn:Rob4}) with stochastic simulations. Below the epidemic threshold (where the standard model switches from zero to an endemic steady-state), the QSD method captures the behaviour reasonably accurately. As $\tau\times\bar{d}/\gamma$ increases, the approximation diverges, with differing levels of performance on each of the graphs tested. This individual-based method performs best on the complete network, on which it provides a good approximation to the expected number of infected individuals for all parameter values. Some level of accuracy is also observed on Zachary's Karate club. However, on the square-lattice this method does not perform well when above the epidemic threshold, significantly overestimating the expected number of infected individuals in the QSD. This is because the structure of the lattice results in significant local correlations which makes the assumption of statistical independence of individual nodes unrealistic. 

Using the population-level individual-based model (see Section~\ref{sec:Indiv} and Appendix~\ref{App:IndivPop}), little accuracy is lost (Figure~\ref{fig:Indiv}(b)). The same pattern of performance occurs across the three networks, and by overlaying the results, the population-level model is almost indistinguishable from the node-level model on the resolution of the graph. This suggests that the QSD approximation is mainly determined by the degree distribution, though there is likely to be some minor variations for graphs with the same distribution but differing in other network properties.

Since the assumption of individual-level statistical independence can lack accuracy, we developed a node-level pair-based model for the QSD (see Section~\ref{sec:Pair} and Appendix~\ref{App:pair}). Figure~\ref{fig:Pair}(a) shows the accuracy of this approximation, which is significantly improved over the individual-based models on all networks. On the complete network and Zachary's karate club, this approximation is very accurate, and on the lattice it loses some accuracy but significantly outperforms the individual-based approximation. The loss of accuracy on the lattice is expected, since pair-approximation methods are generally considered to perform weakly on such structures.

Although the pair-based model is computationally feasible, for large graphs it can be slow. Therefore, we derived a population-level pair-based model  (see Section~\ref{sec:Pair} and Appendix~\ref{App:population}). Again, little accuracy is lost for all networks (Figure~\ref{fig:Pair}(b)), with the result being indistinguishable from the node-level model.

For each of the methods proposed, a stationary solution is reached for all parameter values on all networks. These solutions appear to be unique and lower bounded by 1. Therefore, the proposed methods satisfy the basic properties of the QSD. Sufficiently above the epidemic threshold, our models and the standard (unconditioned) models coincide (Figures~\ref{fig:Indiv} and~\ref{fig:Pair}), showing that the standard models approximate the expected number of infected individuals in the QSD in this region. However, as the transmission rate decreases, the steady states of the standard models deviate from this, eventually tending to the disease-free steady-state. Therefore, the standard models are not a reliable measure of the expected prevalence in the QSD since they do not capture this for all parameter values, and the endemic steady-state in the intermediate range (between the disease-free steady-state and coinciding with the QSD model) is hard to relate to any properties of the underlying stochastic process. The models we propose are more robust for providing insight into the stochastic epidemic model.
\begin{figure}[h]
\hspace{-1.8cm}
\includegraphics[width=1.2\textwidth]{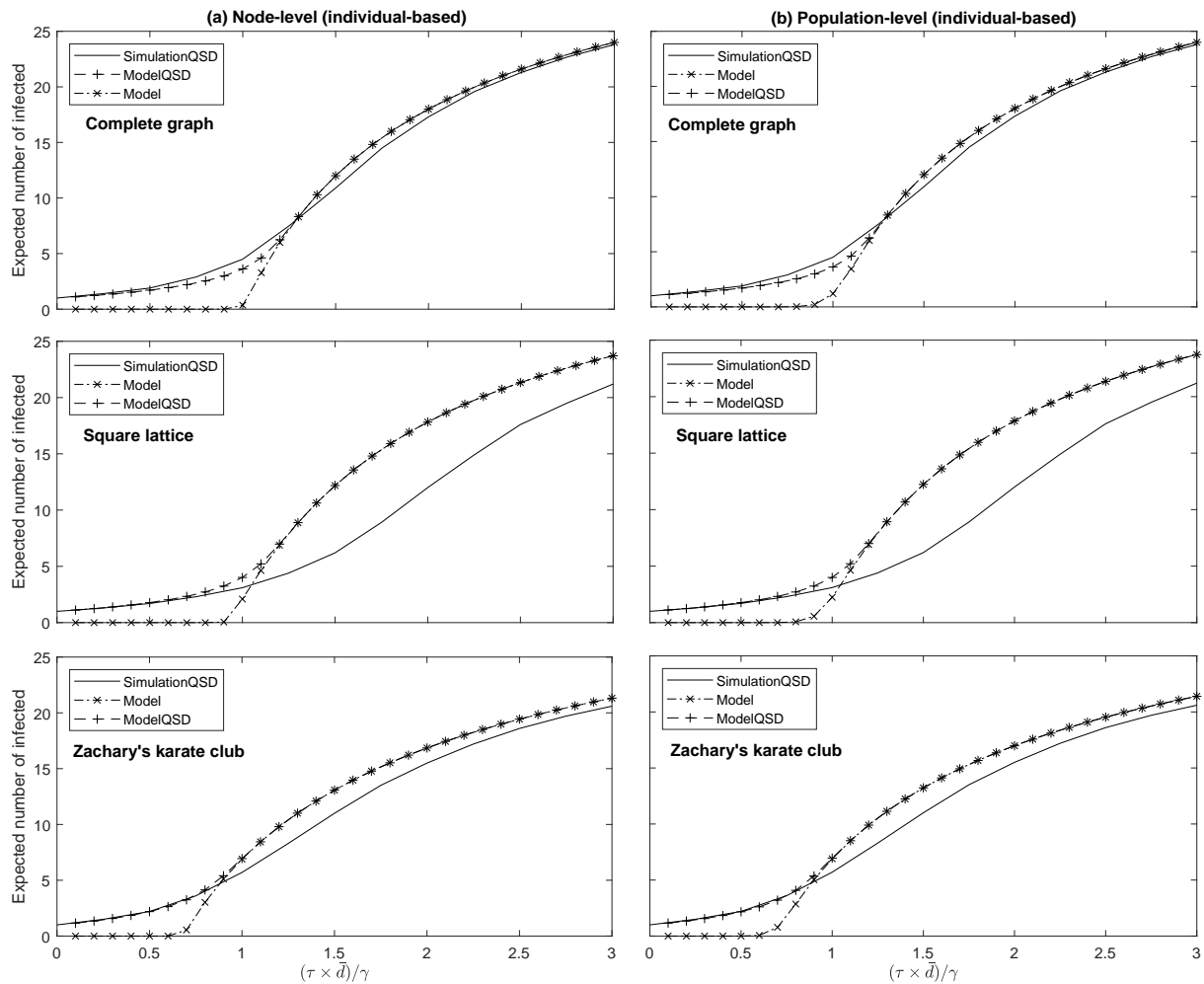}
\caption{The expected number of infected individuals in the QSD as calculated by the individual-based model versus stochastic simulation on a 36 node complete network, 36 node (6x6) square-lattice and the 34 node karate club network, for a range of parameter values. The left plot shows the node-level methods (Equation~\eqref{eqn:Rob4}) and the right shows the population-level methods (Equation~\eqref{eqn:PopIndiv5}). The solid lines represent the average of 10,000 stochastic simulations conditioned against extinction, the dashed line (plusses) represents the proposed QSD approximation method and the dash-dotted line (crosses) represent the standard unconditioned model. The simulated QSD is accurate to within the resolution of the line.}
\label{fig:Indiv}
\end{figure}

\begin{figure}[h]
\hspace{-1.8cm}
\includegraphics[width=1.2\textwidth]{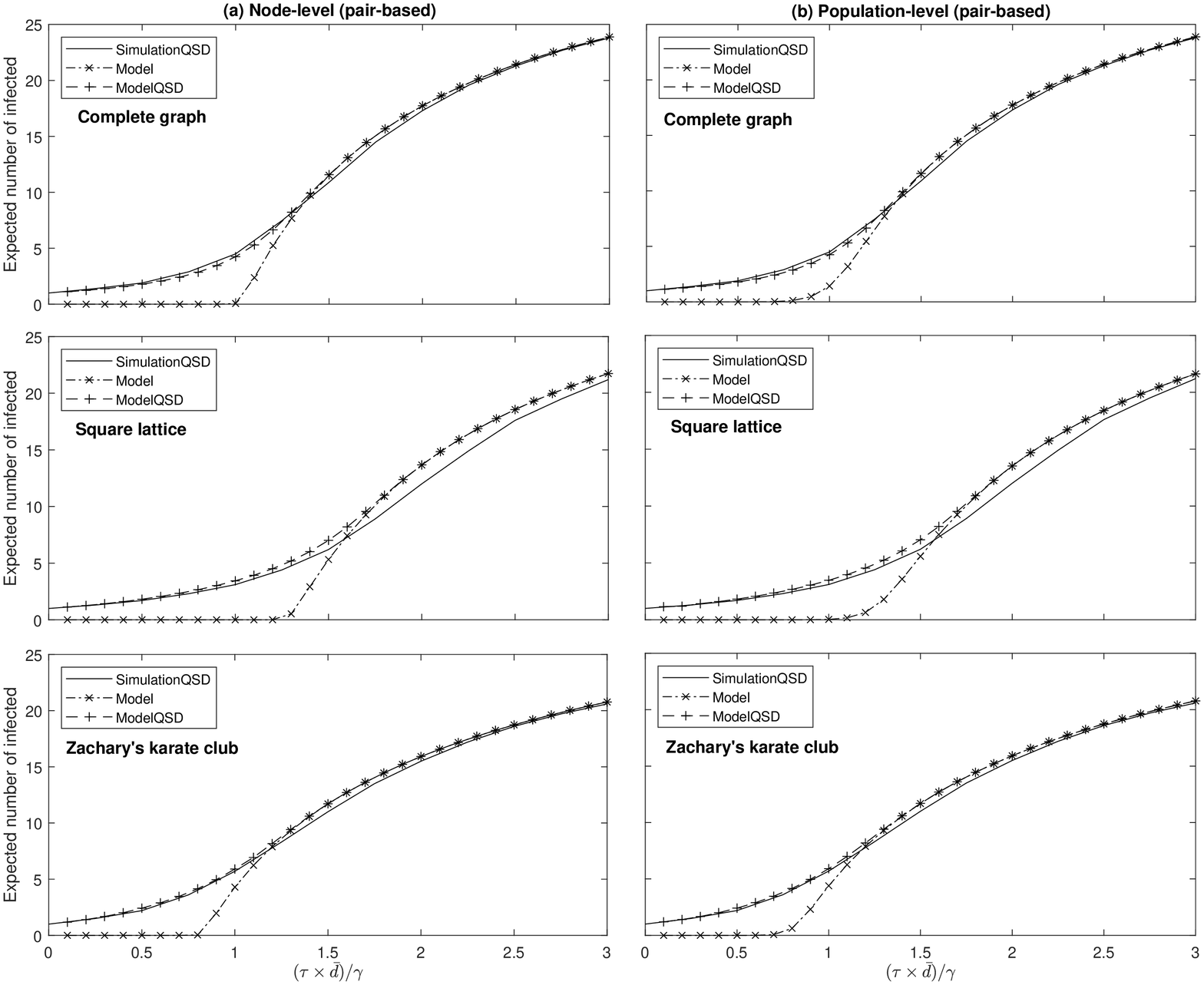}
\caption{The expected number of infected individuals in the QSD as calculated by the pair-based model versus stochastic simulation on a 36 node complete network, 36 node (6x6) square-lattice and the 34 node karate club network, for a range of parameter values. The left plot shows the node-level methods (Equation~\eqref{eqn:PairNode5}) and the right shows the population-level methods (Equation~\eqref{eqn:PairPop5}). The solid lines represent the average of 10,000 stochastic simulations conditioned against extinction, the dashed line (plusses) represents the proposed QSD approximation method and the dash-dotted line (crosses) represent the standard unconditioned model. The simulated QSD is accurate to within the resolution of the line.}
\label{fig:Pair}
\end{figure}
%
%\begin{figure}[H]
%\hspace{-1.8cm}
%\includegraphics[width=1.2\textwidth]{PopvsnodeZachary.eps}
%\caption{Node-level approximations compared to population-level approximations. The left plot shows the individual-based methods and the right shows the pair-based methods. The solid lines represent the average of 10,000 stochastic simulations conditioned against extinction, the dashed line represents the conditioned method and the dash-dotted line represent the unconditioned model. Dashed (and dash-dotted) lines with markers are solutions to the population-level models, and without markers are the node-level models.}
%\label{fig:Comparison}
%\end{figure}

\subsection{Impact of network size}
We now investigate how increasing the size of the population affects the accuracy of the results, testing a 100 node (10x10) lattice, 225 node (15x15) lattice and 400 node (20x20) lattice. Here the square-lattice is chosen because this presented itself as the worst case, with other networks expected to perform better. The lattice is expected to perform badly because the strict structure leads to very high local correlations, which may not be captured by the moment-closure approximations. 

Since the population-level models perform similarly to the node-level models at capturing the expected number of infected individuals, with significantly reduced computational cost, in this section we only use these models to approximate the dynamics. Comparing the QSD method to the simulation results~(Figure~\ref{fig:Size}), we see good agreement for low transmission parameters for both individual-based and pair-based methods. However, once the individual-based methods pass the epidemic threshold, where the standard method reaches a non-zero steady state, both the standard method and the QSD method diverge significantly from the simulation results, overestimating the true expected number of infected in the QSD, echoing what we observed in Figure~\ref{fig:Indiv}. For the pair-based models, once the parameters exceed the epidemic threshold, we still see some deviation from the simulation results for both the standard and QSD methods. However, this is much smaller than for the individual-based methods. For higher relative transmission rates, the model solutions provide a reasonable approximation to the expected number of infected individuals in the QSD. For the three lattice sizes considered, in the regions below and sufficiently above the epidemic threshold, the relative magnitude of the discrepancy between the approximations and simulation results does not change with population size, for both individual-based and pair-based models. However, in the intermediate region there is some sensitivity to population size. Below the epidemic threshold, the standard models do not capture the dynamics of the QSD, regardless of population size, whereas the QSD approximation models are accurate.

\begin{figure}[h]
\hspace{-1.8cm}
\includegraphics[width=1.2\textwidth]{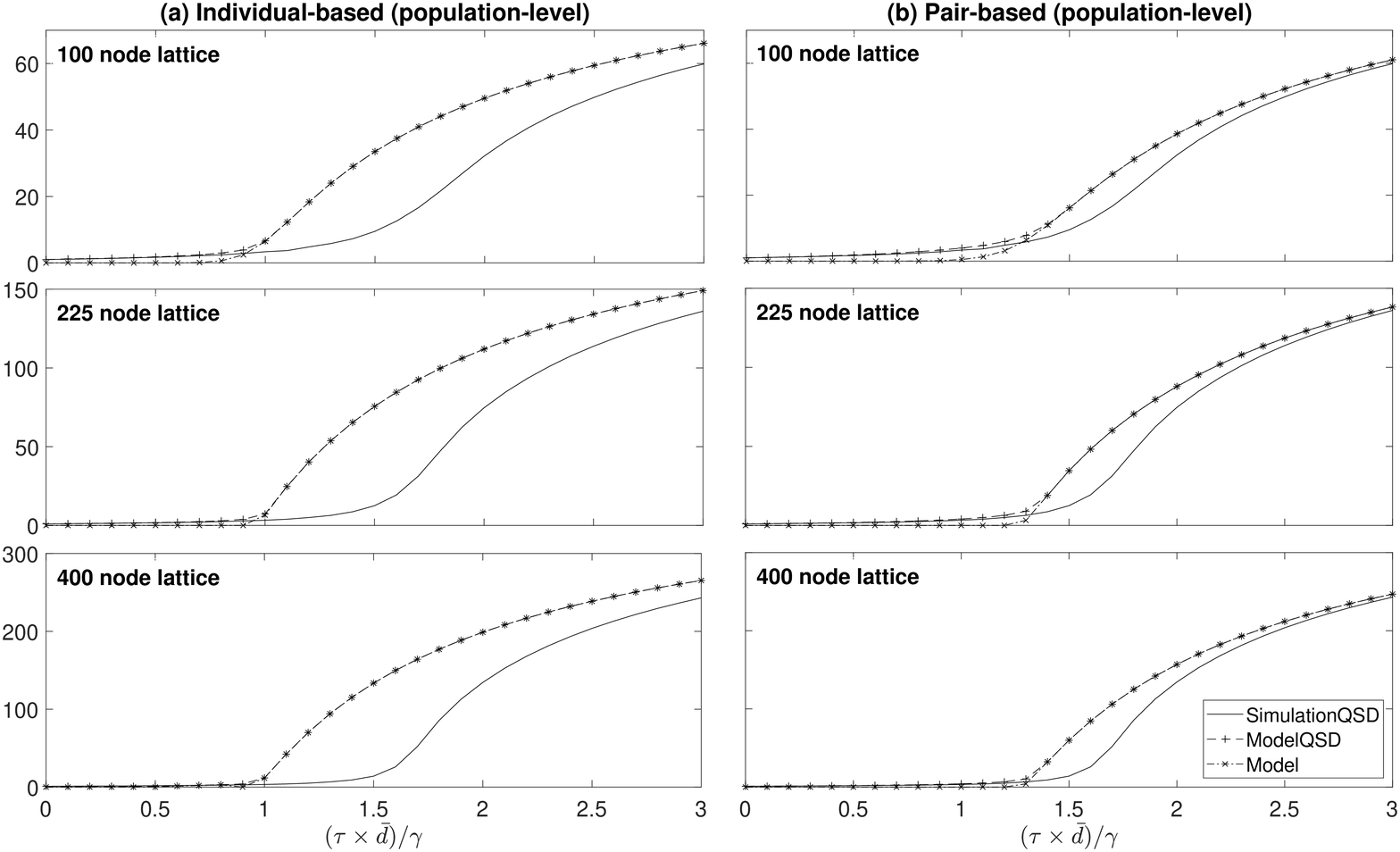}
\caption{The expected number of infected individuals in the QSD as calculated by the population-level models versus stochastic simulation on a 10x10 node square-lattice, 15x15 node square-lattice, and 20x20 node square-lattice for a range of parameter values. The left plot shows the individual-based methods (Equation~\eqref{eqn:PopIndiv5}) and the right shows the pair-based methods (Equation~\eqref{eqn:PairPop5}). The solid lines represent the average of 100,000 stochastic simulations conditioned against extinction, the dashed lines (plusses) represent the proposed QSD approximation method and the dash-dotted lines (crosses) represent the standard unconditioned model. The simulated QSD is accurate to within the resolution of the line.}
\label{fig:Size}
\end{figure}

\section{Discussion}
The standard deterministic SIS model~\cite{Lajmanovich1976,vanMieghem2011,vanMieghem2009,Wang2003} exhibits an epidemic threshold below which the pathogen will go extinct and above which the pathogen will reach an endemic steady-state solution~\cite{Lajmanovich1976,vanMieghem2011}. More complicated `deterministic' models have been developed, such as pair-approximations models~\cite{Frasca2016,Hadjichrysanthou2015,Keeling1999,Keeling2005,Kissetal2017,Mata2013,Sharkey2011}, in which this threshold behaviour is also observed~\cite{Keeling1999,Mata2013}. However, no steady-state solution exists in the stochastic SIS model, making it hard to relate the deterministic and stochastic models in finite populations. 

When the time to absorption (extinction of the pathogen) is long enough, the endemic equilibrium of the network-based deterministic SIS models has been observed to approximate the expected number of infected individuals over short enough time scales. This is quantified by the quasi-stationary distribution (QSD) of the stochastic models~\cite{Allen2000,Kissetal2017,vanMieghem2011}. Although the unconditioned models numerically approximate the expected prevalence in the QSD when sufficiently above threshold, this is not the case as the system approaches the threshold (and below threshold). In particular, comparison of a model with a genuine steady state with one without such a state is not well-defined. To correct this discrepancy, we constructed approximate models which are formally related to the stochastic dynamics via the QSD. The resulting models capture properties of the QSD at all levels of epidemic severity.

Our first approach assumed that the states of neighbouring nodes are independent, at both node-level and population-level. Although this assumption is not particularly realistic, on the complete network this provides a good approximation to the expected number of infected individuals in the QSD (Figure~\ref{fig:Indiv}). However, for more structured networks the accuracy decreased. Little accuracy was lost when computing the expected number of infected individuals using the population-level model compared to the node-level model, with a significant reduction in computational cost. To improve accuracy, we developed node-level and population-level methods based on assuming independence at the level of pairs, which performed well on all networks tested (Figure~\ref{fig:Pair}). Again, little accuracy was lost in the population-level model. With the significant reduction in computational cost, the population-level models are therefore superior to the node-level models for capturing the expected prevalence in the QSD. However, one advantage of the node-level models is the insight these can give into the dynamics of individual nodes in the population, which the population-level models lose.

With the standard unconditioned approximation methods, it is not inherently clear what the models are capturing, since the stochastic model does not exhibit a stable steady-state. By developing conditioned approximation models that capture the quasi-stationary distribution of the stochastic model, we have presented an approximation framework that is directly related to the underlying stochastic process. Sufficiently above the epidemic threshold, the unconditioned standard models coincide with the conditioned QSD models, demonstrating, as expected, that the standard models approximate the QSD when above threshold. Through directly approximating the QSD, the conditioned models are consistent in approximating the QSD for all parameter values. This consistency makes them a more robust method for capturing quasi-stationary behaviour of stochastic epidemic models.

This paper has focussed on the theoretical insights this model grants, and we have shown that the models can be reasonably accurate on a variety of networks. In particular, we show that the pair-based model can perform well on a square-lattice, which is expected to be one of the worst cases for moment-closure approximations. The accuracy and deterministic nature of the models makes them more amenable to analysing how different network structures can alter the statistics of the QSD than the use of stochastic simulation. This is valuable for characterising the likelihood and severity of the epidemic, for example through the invasion probability~\cite{Wilkinson2013}, which can be calculated directly from the node-level models proposed, and the expected prevalence, which we presented in the numerical results. The potential future applications of this work include applying the methods to investigate how network structure, such as the degree variance, affects the properties of the QSD, as well as extending the model to other epidemic and population dynamics models.

\begin{appendices}

\section{Node-level individual-based QSD model}
\renewcommand{\thefigure}{A\arabic{figure}}
\setcounter{figure}{0} 
\renewcommand{\theequation}{A\arabic{equation}}
\setcounter{equation}{0} 
\subsection{Derivation of node-level conditional distribution equation}
\label{App:nodelevel}
The rate of change in the probability that node $i$ is infected in the QSD is given by the sum of the rates of change in the full system state probabilities for which node $i$ is infected. That is, we have
\begin{eqnarray}
\frac{\mathrm{d}}{\mathrm{d}t}\left(\rho(\Sigma_i(t)=I)\right)&=&\sum_{\alpha:\sigma_{\alpha i}=I}\frac{\mbox{d} \rho_\alpha}{\mbox{d} t} \nonumber \\
&=& \frac{\sum_{\alpha:\sigma_{\alpha i}=I}(QP)_\alpha}{1- P_1} + \frac{(QP)_1}{(1-P_1)^2}\sum\limits_{\alpha:\sigma_{\alpha i}=I}P_\alpha , 
\label{eqn:NodeCond}
\end{eqnarray}
where the terms are defined in Section~\ref{sec:QSD}.
The numerator of the first term on the second line corresponds to the rate of change in the probability that node $i$ is infected, which is given by $\langle \dot{I_i} \rangle $ in Equation~\eqref{eqn:StandardIndiv} in Appendix~\ref{sec:Standard}. The summation in the second term corresponds to the probability that node $i$ is infected, $\langle I_i \rangle$. Therefore, we can write
\begin{eqnarray}
\frac{\mathrm{d}}{\mathrm{d}t}\left(\rho(\Sigma_i(t)=I)\right)= \frac{\langle \dot{I_i}\rangle}{1- P_1} + \frac{(QP)_1}{(1-P_1)^2}\langle I_i \rangle.\nonumber
\end{eqnarray}
Here $(QP)_1$ is the rate at which the system enters the absorbing state. The system can only reach the absorbing state from a state with a single infected individual, in node $j$ for example, which transitions to the all susceptible state at rate $\gamma_j$. Therefore $(QP)_1=\sum\limits_{j}\gamma_j\langle I_j S \rangle$, where we use $\langle I_j S \rangle$ to denote the probability that node $j$ is infected and all other nodes are susceptible. Using this along with Equation~\eqref{eqn:StandardIndiv}, we obtain
\begin{equation}
\frac{\mathrm{d}}{\mathrm{d}t}\left(\rho(\Sigma_i(t)=I)\right)= \frac{\sum_{j}{T_{ij}} \langle S_iI_j \rangle-\gamma_i \langle I_i \rangle}{1- P_1} + \frac{ \langle I_i \rangle }{(1-P_1)^2}\sum\limits_{j}\gamma_j\langle I_j S \rangle.
\label{eqn:IndivCondApp}
\end{equation}
\subsection{Proof that the individual-based node-level QSD model is invariant on $[0,1]^N$.}
\label{App:invariant}
\begin{proof}
To prove that the model in Equation~\eqref{eqn:Rob4} is invariant we use the method from~\cite{Lajmanovich1976}. Along the boundaries to the set we are interested in, we either have $\langle Y_i \rangle=0$ and $\langle X_i \rangle =1$ or $\langle Y_i \rangle=1$ and $\langle X_i \rangle =0$. To show the system is invariant, we need to show that along these boundaries the trajectories do not point away from this set. 

First consider $\langle Y_i \rangle=0$. At this boundary, we have
\begin{equation}
\langle \dot{Y_i} \rangle = \sum\limits_j T_{ij}\langle Y_j \rangle.
\end{equation}
If $\langle Y_j \rangle \in [0,1]$, this cannot be negative, and therefore at $\langle Y_i \rangle =0$ the trajectory in the $i$ direction cannot leave the set $[0,1]^N$. Now consider $\langle Y_i\rangle=1$. We have
\begin{equation}
\langle \dot{Y_i} \rangle = -\gamma_i + \gamma_i\prod\limits_{k \neq i}\langle X_k\rangle.
\end{equation}
The product in this equation is in $[0,1]$ if $\langle X_k \rangle \in [0,1]$ for all $k$. Therefore, this equation can never be positive, so along this boundary the trajectory cannot leave the set $[0,1]$. Therefore, this model is invariant on $[0,1]^N$. \end{proof}

\subsection{Proof of Theorem 1}
\label{App:existence}
\begin{proof}
Consider the node-level individual-based model (Equation~\eqref{eqn:Rob4}) on a $k$-regular network with homogeneous transmission and recovery. If we start with a fully infected population, $\langle Y_i \rangle$ will be equal for all $i$ at every time point. Therefore we can denote $\langle I_i \rangle = a$ for all $i \in \mathcal{V}$. We can write the rate of change in the node probabilities as
\begin{equation}
\dot{a}=-\gamma a + \tau ka(1-a) + a\frac{\gamma Na(1-a)^{N-1}}{1-(1-a)^N}.
\label{eqn:app1}
\end{equation}
In the steady state $\dot{a}=0$. If we rule out $a=0$, since Equation~\eqref{eqn:app1} is undefined for $a=0$, then we obtain
\begin{equation}
(1-a)k\left(\frac{\tau}{\gamma}+\frac{N}{k}\frac{a(1-a)^{N-2}}{1-(1-a)^N}\right)=1.
\end{equation}
We are therefore interested in solutions to $f(a)=0$ with $a \in [0,1]$, where 
\begin{equation}
f(a)=(1-a)k\left(\frac{\tau}{\gamma}+\frac{N}{k}\frac{a(1-a)^{N-2}}{1-(1-a)^N}\right)-1.
\end{equation}
To see if a solution exists within this interval we check the signs at the end points. 

At $a=1$
\[f(1)=-1<0\]
the function is negative.

As $a$ goes to zero
\[\lim\limits_{a \to 0}f(a)=k\frac{\tau}{\gamma}-1 +\lim\limits_{a \to 0}N\frac{a(1-a)^{N-1}}{1-(1-a)^N}.\]
\[\lim\limits_{a \to 0}N\frac{a(1-a)^{N-1}}{1-(1-a)^N}=\lim\limits_{a \to 0}N\frac{(1-a)^{N-1}+(N-1)a(1-a)^{N-2}}{N(1-a)^{N-1}}=1.\]
\[\implies \lim\limits_{a \to 0}f(a) = k\frac{\tau}{\gamma}> 0 \text{ if } \frac{\tau}{\gamma}>0.\]
Therefore as long as the transmission rate $\tau$ is greater than zero there exists a solution to $f(a)=0$ in the open interval $(0,1)$, since $f(a)$ is non-singular on $(0,1)$. 

We now need to show that our approximation to the expected number of infected individuals in the QSD is bounded below by one. This proof holds for all networks provided a solution exists satisfying $\langle Y_i \rangle \in (0,1)$ for all $i$, which we have proven for $k$-regular networks. Consider the node-level individual-based model; i.e.
\begin{equation}
\langle \dot{Y_i}\rangle=-\gamma_i\langle Y_i\rangle + \sum\limits_j T_{ij}\langle X_i\rangle \langle Y_j\rangle +\frac{\langle Y_i \rangle}{1-\prod\limits_k\langle X_k\rangle} \sum\limits_j\gamma_j\langle Y_j\rangle \prod\limits_{k\neq j} \langle X_k\rangle
\label{eqn:AppNodeLevel}
\end{equation}
To approximate the QSD we calculate $\langle Y_i^*\rangle/(1-\prod\limits_k \langle X_k^* \rangle)$, where $\langle Y^*\rangle$ and $\langle X^*\rangle$ are steady-state solutions to~\eqref{eqn:AppNodeLevel}.

Let S be the sum of N independent Bernoulli random variables with success probabilities given by the vector $\langle Y_i^* \rangle$ for  $i \in \{1,2,….,N\}$, which is a feasible solution of Equation~\eqref{eqn:AppNodeLevel}. It is straightforward then that $\mathbb{E}[S]=\sum_i\langle Y_i^*\rangle$, and we can write
\begin{align}
\sum_i \langle Y_i^*\rangle =& \mathbb{E}[S] = \sum_{x=1}^{x=N} \mathrm{P}(S=x)x \geq \sum_{x=1}^{x=N} \mathrm{P}(S=x) = \mathrm{P}(S \geq 1) \\
=& 1- \prod_j(1-\langle Y_j^*\rangle)
\end{align}
So when we approximate the expected number infected in the QSD as
\begin{equation}
\frac{\sum_i \langle Y_i^*\rangle}{1-\prod_j(1- \langle Y_j^*\rangle)}
\end{equation}
this cannot be less than 1. Therefore, provided a non-zero solution exists to Equation~\eqref{eqn:AppNodeLevel}, the approximation to the expected number of infected individuals in the QSD is not less than 1. \end{proof}

\section{Standard approximate models}
\label{sec:Standard}
\renewcommand{\thefigure}{B\arabic{figure}}
\setcounter{figure}{0} 
\renewcommand{\theequation}{B\arabic{equation}}
\setcounter{equation}{0} 
Due to the prohibitive computational cost of solving the master equation (Equation~\eqref{eqn:ME}), approximation methods are useful. In this section, we give an overview of the heterogeneous mean-field and pair-approximation methods, which can be interpreted as approximating the expected behaviour of the stochastic model. For detailed derivations and analysis of these models see~\cite{Kissetal2017}.

Under the heterogeneous mean-field model, we assume that: all individuals with the same degree can be treated identically, the status of neighbouring individuals are independent, $\gamma_i=\gamma$ for all $i \in \mathcal{V}$, and $T_{ij}=\tau$ for all $i,j \in \mathcal{V}$ with $T_{ij}>0$ or $T_{ji}>0$ (the network is assumed undirected for simplicity). The rate of change in the expected number of susceptible and infected individuals, stratified by the degree of the individual, is then approximated by~\cite{Kissetal2017}
\begin{align}
\dot{[S_k]} &\approx -\tau \sum\limits_{l \in \mathcal{M}}|C_{k,l}|\frac{[S_k]}{|C_k|}\frac{[I_l]}{|C_l|}+\gamma[I_k] \nonumber\\
\dot{[I_k]} &\approx \tau \sum\limits_{l \in \mathcal{M}}|C_{k,l}|\frac{[S_k]}{|C_k|}\frac{[I_l]}{|C_l|}-\gamma[I_k],
\end{align}
where $[S_k]$ is the expected number of susceptible individuals of degree $k$ at time $t$, $|C_k|$ is the number of degree $k$ nodes, $|C_{k,l}|$ is the number of pairs involving a degree $k$ node and a degree $l$ node, and $\mathcal{M}$ is the set of unique degrees on the network. Above, and throughout, we use `dot' notation for derivatives with respect to time. Whilst the assumption of neighbouring individuals being independent is unrealistic, the resulting model has low computational cost, and hence it is popular to study.

Instead of assuming statistical independence between individuals, models have been derived by writing down exact equations for the expected number of individuals and pairs:
\begin{align}
[\dot{S_k}]=&\gamma[I_k]-\sum\limits_{l\in\mathcal{M}}\tau[S_kI_l] \nonumber\\
[\dot{I_k}]=&-\gamma[I_k]+\sum\limits_{l\in\mathcal{M}}\tau[S_kI_l] \nonumber\\
[\dot{S_kI_l}]=& \gamma([I_kI_l]-[S_kI_l])+\tau(\sum\limits_{m\in \mathcal{M}}[S_kS_lI_m]-\sum\limits_{m\in \mathcal{M}}[I_mS_kI_l]-[S_kI_l])\nonumber\\
[\dot{S_kS_l}]=& \gamma([S_kI_l]+[I_kS_l])-\tau(\sum\limits_{m\in \mathcal{M}}[S_kS_lI_m]+\sum\limits_{m\in \mathcal{M}}[I_mS_kS_l])\nonumber\\
[\dot{I_kI_l}]=&\tau([S_kI_l]+[I_kS_l]-2\gamma[I_kI_l]+\tau(\sum\limits_{m\in \mathcal{M}}[I_mS_kI_l]+\sum\limits_{m\in \mathcal{M}}[I_kS_lI_m]),
\label{eqn:PairPopStandard}
\end{align}
where $[A_kB_l]$ is the expected number of pairs at time $t$, between degree $k$ and $l$ individuals in states $A$ and $B$ respectively, and $[A_kB_lC_h]$ is the expected number of triples at time $t$, between degree $k$, $l$ and $h$ individuals, in states $A$, $B$ and $C$ respectively. 

Solving this system exactly involves deriving a full hierarchy of equations describing triples and quads and so on~\cite{Eames2002}, and therefore we wish to approximate this system by closing the hierarchy early. This can be done by expressing triples as some function of pairs and individuals. To approximate the triples, we analyse the number of edges starting from a susceptible node, following~\cite{Eames2002,Kissetal2017}. The total number of $SA$ edges (for $A\in \{S,I\}$) from a degree $k$ node to a degree $l$ node are $[S_kA_l]$. Since we have $[S_k]$ susceptible degree $k$ nodes, we have approximately $[S_kA_l]/(k[S_k])$ edges leading from a given susceptible degree $k$ node to a given degree $l$ node in state $A$. Therefore, for a chosen susceptible degree $k$ node the probability that two neighbours, with degree $l$ and $m$, are in states $A$ and $B$ is given by $[A_lS_k][S_kB_m]/k^2[S_k]^2$. We have $k(k-1)$ choices of the two neighbours, and $[S_k]$ choices of the susceptible node, and therefore we can approximate the expected number of triples $[A_lS_kB_m]$ as 
\begin{equation}
[A_lS_kI_m] \approx \frac{k-1}{k}\frac{[A_lS_k][S_kI_m]}{[S_k]}.
\label{eqn:triple1}
\end{equation}
This approximation makes the homogeneity assumption that the neighbours of susceptible degree $k$ nodes are interchangeable and the states of pairs are independent. Using this expression, the system of equations~\eqref{eqn:PairPopStandard} is closed at the level of pair terms, which allows the system to be solved with reasonably low computational cost. 

These two models act at the population level, since they describe how the expected number of individuals with certain traits change. Following the motivation behind these models, node-level models have been developed that describe how the probability of individual nodes being infected change with time. Such models have been referred to as individual-based models~\cite{Sharkey2011,Sharkey2015}, node-level models~\cite{Overton2019}, propagation models~\cite{Kissetal2017} or quenched-mean field~\cite{Ferreiraetal2012,Mata2013}. The advantage of such models over the population-level models is that we do not need to make any homogeneity assumptions about the underlying populations, and therefore properties such as clustering, directed edges and degree heterogeneity are naturally captured. The downside however is that the computational cost scales with at least the number of nodes.

Under Markovian network-based SIS, the dynamics of individual nodes are given by~\cite{Sharkey2011} 
\begin{align} \nonumber
\dot{\langle S_i \rangle}&=-\sum_{j}{T_{ij}} \langle S_iI_j \rangle+\gamma_i \langle I_i \rangle ,\\
\dot{\langle I_i \rangle}&=\sum_{j}{T_{ij}} \langle S_iI_j \rangle-\gamma_i \langle I_i \rangle ,
\label{eqn:StandardIndiv}
\end{align}
where $\langle A_i \rangle$ represents the probability $P(\Sigma_i(t)=A)$ with $A \in \{S,I\}$, and $\langle A_i B_j \rangle$ represents the probability $P(\Sigma_i(t)=A, \Sigma_j(t)=B)$ with $A,B \in \{S,I\}$. 

This equation exactly describes the rate of change for individual nodes in terms of pairs. Pairs of nodes are exactly described by
\begin{align}\label{eqn:StandardPair} \nonumber
\dot{\langle S_iI_j \rangle}=&\sum_{k}{T_{jk}}\langle S_iS_jI_k \rangle-\sum_{k}{T_{ik}}\langle I_kS_iI_j \rangle \\ \nonumber
& -(T_{ij}+\gamma_j )\langle S_iI_j \rangle + \gamma_i \langle I_iI_j \rangle ,\\ \nonumber
\dot{\langle S_iS_j \rangle}=& -\sum_{k}{T_{jk}} \langle S_iS_jI_k \rangle -\sum_{k}{T_{ik}} \langle I_kS_iS_j \rangle, \\ \nonumber
\dot{\langle I_iI_j \rangle}=&\sum_{k}{T_{jk}}\langle I_iS_jI_k \rangle +\sum_{k}{T_{ik}}\langle I_kS_iI_j \rangle -(\gamma_i +\gamma_j )\langle I_iI_j \rangle \\
& + T_{ij}\langle S_iI_j \rangle + 
T_{ji} \langle I_iS_j \rangle, \\ \nonumber
\end{align}
where $\langle A_i B_j C_k \rangle$ represents the probability $P(\Sigma_i(t)=A, \Sigma_j(t)=B, \Sigma_k(t)=C)$ with $A,B,C \in \{S,I\}$. To solve this requires a hierarchy of equations up to full system size. Following similar logic to the population-level equations, this system can be approximated by making assumptions of statistical independence. Assuming that the states of individuals are independent, $\langle S_iI_j\rangle \approx \langle S_i \rangle \langle I_j \rangle$, we can close the hierarchy at the level of individuals.
Alternatively, we can assume independence at the level of pairs. The natural assumption of statistical independence to apply to pairs is that, given three nodes in a line, if the state of the central node is known then the state of the outer two nodes are independent. For all triples in the system above, the central node in the configuration is always the centre node of a line between the two outer nodes. Therefore, if we consider the triple $\langle A_i B_j C_k \rangle$, this can be approximated as a function of lower order terms by using conditional probabilities and assuming statistical independence. By the definition of conditional probabilities, we obtain
\begin{equation}\langle A_i B_j C_k \rangle= {\langle A_iC_k|B_j \rangle}{\langle B_j \rangle}.\nonumber\end{equation}
Assuming that the states of nodes $i$ and $k$ are independent given the state of node $j$, this becomes
\begin{equation}{\langle A_iB_jC_k \rangle} \approx {\langle A_i|B_j\rangle \langle C_k|B_j \rangle}{\langle B_j \rangle}= \frac{\langle A_iB_j\rangle \langle B_jC_k \rangle}{\langle B_j \rangle},\label{eqn:closure}\end{equation}
which closes the hierarchy at the level of pairs. Other methods to approximate triples in terms of pairs and individuals have been proposed~\cite{Keeling1999,Rogers2011,Sharkey2011}, however we do not consider them in this paper. 

The population-level methods described above can be derived rigorously from the node-level methods~\cite{Sharkey2011}. In the exact case, we have
\begin{equation}
[A_k]=\sum\limits_{j:k_j=k}\langle A_j\rangle
\end{equation}
and
\begin{equation}
[A_kB_l]=\sum\limits_{i:k_i=k}\sum\limits_{j:k_j=l}\langle A_iB_j\rangle
\end{equation}
where $A,B \in \{S,I\}$ and $k_i$ is the degree of node $i$. Using this, the rate of change for the population-level terms can be derived. From this, we can also approximate the node-level quantities as
\begin{align}
\langle A_i \rangle \approx \frac{[A_{k_i}]}{|C_{k_i}|},
\label{eqn:PopIndivRelation}
\end{align}
and
\begin{equation}
\langle A_iB_j \rangle \approx \frac{[A_{k_i}B_{k_j}]}{|C_{k_i,k_j}|}.
\label{eqn:PopPairRelation}
\end{equation}
The models described here exhibit an epidemic threshold, above which the pathogen persists and below which the pathogen dies out (illustrated in Figure~\ref{fig:Standard} for the node-level pair-based model). For the population-level models and individual-based node-level model, above these thresholds a unique, globally stable steady-state exists~\cite{Keeling1999,Keeling2005,Kissetal2017,Lajmanovich1976,vanMieghem2011}. For the node-level pair-based model, the disease-free solution has been shown to become unstable as the transmission rate increases~\cite{Mata2013}, at which point we have shown that an endemic steady-state solution exists (Appendix~\ref{App:standard}). Numerically, this endemic equilibrium appears to be unique and globally attracting, similar to the endemic solutions in the other models.

\begin{figure}[H]
\includegraphics[width=1\textwidth,width=1\textwidth]{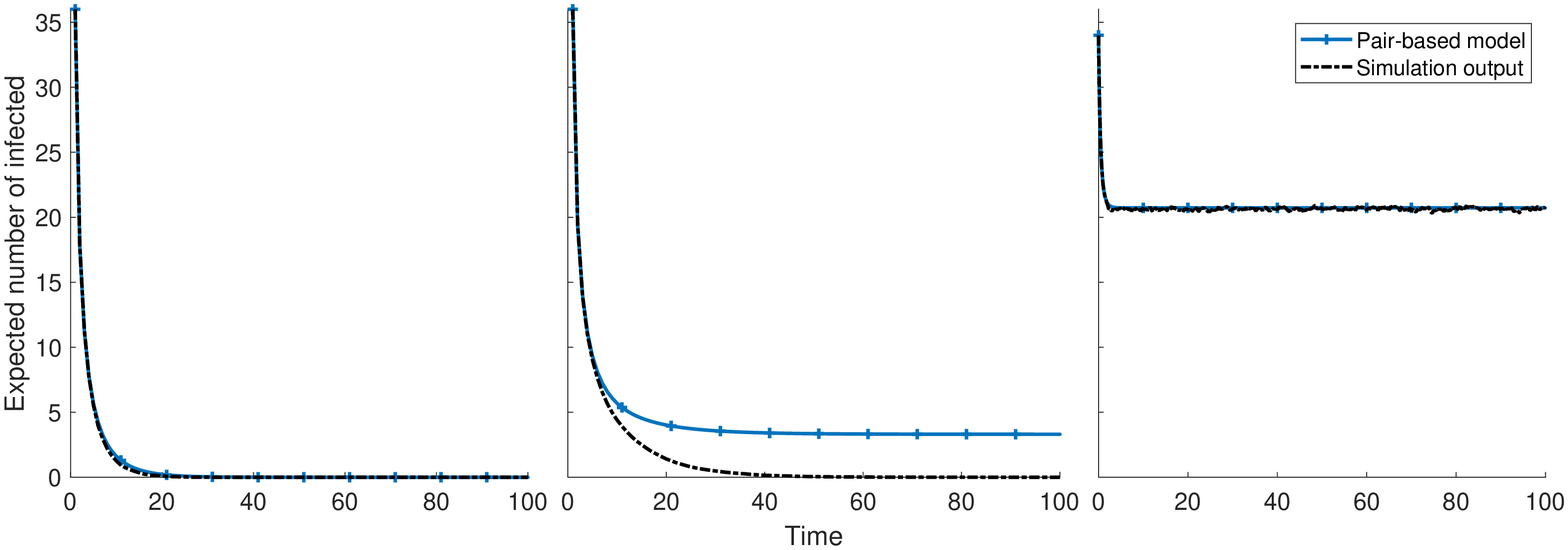}
\caption{Comparing the standard pair-based model (equations~\eqref{eqn:StandardIndiv} and~\eqref{eqn:StandardPair} with the closure from Equation~\eqref{eqn:closure}) with closures with the output of stochastic simulations on Zachary's karate club network. We plot the expected number of infected individuals against time for each of the methods. As the figures move from left to right the transmission rate increases. In the right-most figure, steady-like behaviour is observed in the stochastic model, since the expected time to extinction is very long.}
\label{fig:Standard}
\end{figure}

When comparing these models to the underlying stochastic process (e.g. Figure~\ref{fig:Standard}), below the epidemic threshold the models accurately capture the expected number of infected individuals in the stochastic process. However, as the transmission rate increases (or recovery rate decreases), we pass the epidemic threshold, and observe an endemic equilibrium that does not correspond to the stochastic process. Eventually, when the parameters are sufficiently above the epidemic threshold, the endemic steady-state solutions of these models can approximate the behaviour of the stochastic model for a long time, since the time to extinction of the pathogen is very long. Here, the stochastic process behaves similarly to the quasi-stationary distribution of the model; i.e. the expected long-term behaviour if extinction has not occurred.

\section{Proof of existence of an endemic steady-state for the standard pair-based model}
\label{App:standard}
\renewcommand{\thefigure}{C\arabic{figure}}
\setcounter{figure}{0} 
\renewcommand{\theequation}{C\arabic{equation}}
\setcounter{equation}{0} 
\begin{proof}
In~\cite{Lajmanovich1976}, a theorem is proven regarding the existence of stable endemic solutions for ordinary differential equation epidemic models. Here we demonstrate that the standard pair-based SIS model (equations~\eqref{eqn:StandardIndiv} and~\eqref{eqn:StandardPair} with the closure from Equation~\eqref{eqn:closure}~\cite{Mata2013}) satisfies the requirements for this proof, and therefore has a stable endemic steady-state.

Consider an ODE of the form
\begin{equation}
\frac{\mathrm{d}y}{\mathrm{d}t}=Ay+N(y).
\label{eqn:proof}
\end{equation}
If the following statements hold, then there exists a threshold above which an endemic steady-state exists.
\begin{enumerate}
\item{A compact convex set $C$ on the domain of $N$ is positively invariant, with $y=0 \in C$.}
\item{$\lim\limits_{y \to 0}||N(y)||/||y||=0$}
\item{There exists $r>0$ and a real eigenvector $w$ or $A^T$ such that $(w \cdot y)\geq r||y|| \ \ \forall y \in C$}
\item{$(w\cdot N(y)) \leq 0 \ \ \forall y \in C$}
\item{$y=0$ is the largest positively invariant set contained in $H=\{y \in C |(w\cdot N(y))=0\}$}
\end{enumerate}
The first step is to write the pair-based model in the form \eqref{eqn:proof}. The pair-based model is given by
\begin{align}
\langle \dot{I_i} \rangle=& \sum\limits_j^NT_{ij}\langle S_iI_j\rangle - \gamma\langle I_i \rangle \\
\langle \dot{S_iI_j} \rangle =& \sum\limits_{k \neq i}^NT_{jk}\frac{\langle S_iS_j\rangle\langle S_jI_k\rangle}{\langle S_j \rangle}-\sum\limits_{k \neq j}^NT_{ik}\frac{\langle I_kS_i\rangle\langle S_iI_j\rangle}{\langle S_i \rangle}- \langle S_iI_j \rangle - \gamma\langle S_iI_j \rangle +\gamma\langle I_iI_j \rangle,
\end{align}
where $\langle S_i \rangle = 1- \langle I_i \rangle$, $\langle I_iI_j \rangle = \langle I_j \rangle - \langle S_iI_j \rangle$ and $\langle S_iS_j \rangle= \langle S_i \rangle - \langle S_iI_j \rangle$.

This can be rewritten as
\begin{align}
\langle \dot{I_i} \rangle=&\sum\limits_j^NT_{ij}\langle S_iI_j\rangle - \gamma\langle I_i \rangle \\
\langle \dot{S_iI_j} \rangle =& - (T_{ij}+2\gamma)\langle S_iI_j \rangle+\gamma\langle I_j \rangle +\sum\limits_{k \neq i}^NT_{jk}\langle S_jI_k\rangle\\
&-\sum\limits_{k \neq i}^NT_{jk}\frac{\langle I_iS_j\rangle\langle S_jI_k\rangle}{\langle S_j \rangle}-\sum\limits_{k \neq j}^NT_{ik}\frac{\langle I_kS_i\rangle\langle S_iI_j\rangle}{\langle S_i \rangle}.
\end{align}
Defining $y_i=\langle I_i \rangle$ for $1\leq i\leq N$ and $y_i=\langle S_1I_{i-N}\rangle$ for $N+1 \leq i \leq 2N$, $y_i=\langle S_2I_{i-2N} \rangle$ for $2N+1 \leq i \leq 3N$, and so on, we can write the pair-based model in the form of Equation~\eqref{eqn:proof}. Compiling the linear terms into the matrix $A$, we see that $A$ is only negative on the diagonal. The remaining non-linear terms define the function $N(y)$, which only assigns negative values to each input. Now it is required to check if the properties hold. 

Property (1.) holds because the system is invariant on the set $C=\{0 \leq \langle I_i \rangle \leq 1; 0 \leq \langle S_iI_j \rangle \leq 1\}$. Property (2.) holds because as $y \to 0$ the denominator of all terms, $1-\langle I_i \rangle$, goes to one, and the numerator is of the form $y_iy_j$, which goes to zero faster than $y_i$ and $y_j$. Property (3.) holds because $A$ is irreducible since all the equations are coupled. Since $A$ is only negative on the diagonal, by the Perron-Frobenius theorem, $A^T$ must have an eigenvector $w$ such that $w_i>0$ for all $i$. Property (4.) holds because the function $N(y)$ is negative, so $(w \cdot N(y))\leq 0$, since $w_i>0$ for all $i$. We now need to test property (5.). 

\noindent\textbf{Property (5.)}
If $ y \in H$ then $(w \cdot N(y))=0$. This implies that
\begin{equation}
w_i\sum\limits_{k \neq i}\frac{T_{jk}\langle I_iS_j \rangle\langle S_jI_k\rangle}{1-\langle I_j\rangle}=0
\end{equation}
and
\begin{equation}
w_i\sum\limits_{k \neq j}\frac{T_{ik}\langle I_kS_i \rangle\langle S_iI_j\rangle}{1-\langle I_i\rangle}=0,
\end{equation}
for all pairs $(i,j)$. If we assume that $y \in H$ and $y \neq 0$, then $y_h \neq 0$ for some $h$. If we assume that $y_h = \langle S_iI_j\rangle \neq 0$, then we must have $\langle S_iI_k\rangle=0$, for all $k \in \mathcal{N}_i$. Also, we require $\langle S_jI_k\rangle=0$ for some $k$ or $\langle I_iS_j\rangle=0$. We now need to investigate whether such a state can be invariant. 

Define $S=\{i:y_i =0\}$ and $S'=\{i:y_i \neq 0\}$, both of which are non-empty since $y \neq 0$ and $\langle S_iI_j\rangle=0$ for some pair $(i.j)$ by the above argument. Since $A$ is irreducible, there must exist a pair $k \in S$ and $h \in S'$ such that $dy_k/dt$ depends on $y_h$. 

First assume that $y_h = \langle S_iI_j\rangle$ and $y_k=\langle I_i \rangle$. We have
\begin{equation}
\frac{\mathrm{d}y_k}{\mathrm{d}t}=\sum\limits_{j\neq i}T_{ij}\langle S_iI_j\rangle
\end{equation}
If this state is invariant, then $dy/dt=0$, which implies that $dy_k/dt=0$ for all $k$. This can only be the case if $\langle S_kI_j \rangle=0$ for all $j$. However, we have assumed that $\langle S_iI_j \rangle \neq 0$, so this is not the case and $dy_k/dt \neq 0$. 

Now assume $y_k=\langle S_jI_i\rangle$, which gives
\begin{eqnarray}
\frac{\mathrm{d}y_k}{\mathrm{d}t}&=& \gamma\langle I_i \rangle +\sum\limits_{m \neq j}^NT_{im}\langle S_iI_m\rangle - \sum\limits_{m \neq j}^NT_{im}\frac{\langle I_jS_i\rangle\langle S_iI_m\rangle}{\langle S_i\rangle}.
\end{eqnarray}
Since $\langle I_jS_i \rangle/\langle S_i \rangle \leq 1$, the sum of the last two terms cannot be negative. Therefore, if $dy_k/dt=0$ we have $\langle I_i \rangle=0$. However, as has been shown by assuming $\langle I_i \rangle=0$, this case is not possible. Therefore, $dy_k/dt \neq 0$. Therefore, if $\langle S_iI_j \rangle \neq 0$ for some pair $(i,j)$ and $y \in H$, then this state cannot be invariant. 

Now assume that $y_h = \langle I_i \rangle \in S'$ for some $i$, and consider $y_k = \langle S_jI_i \rangle \in S$. Since $\langle S_xI_y\rangle =0$ for all $(x,y)$, we have
\begin{eqnarray}
\frac{\mathrm{d}y_k}{\mathrm{d}t}&=& \gamma\langle I_i \rangle.
\end{eqnarray}
Since $\langle I_i \rangle \in S'$, $dy_k/dt \neq 0$. Therefore, there are no invariant sets in $H$ such that $y \neq 0$, and $y=0$ is the largest positively invariant set in $H$.

This shows that properties 1-5 are satisfied for this model. Therefore, there exists a stable endemic steady-state above the epidemic threshold of the standard pair-based SIS model. \end{proof}

\section{Population-level individual-based QSD model}
\label{App:IndivPop}
\renewcommand{\thefigure}{D\arabic{figure}}
\setcounter{figure}{0} 
\renewcommand{\theequation}{D\arabic{equation}}
\setcounter{equation}{0} 
The node-level equations give detailed insight into the dynamics of individual nodes in the QSD, however the number of equations scales with $N$. To build approximations with a reduced number of equations, population-level models can be constructed for undirected networks. The rate of change in the expected number of infected individuals with a given degree, under the conditional distribution, is found by taking the sum over the probability that each node with this degree is infected
\begin{eqnarray}
\sum_{i: k_i=k}\sum_{\alpha:\sigma_{\alpha i}=I}\frac{\mbox{d} \rho_\alpha}{\mbox{d} t}&=& \sum_{i: k_i=k}\left(\frac{\sum_{j}{T_{ij}} \langle S_iI_j \rangle-\gamma_i \langle I_i \rangle}{1- P_1} + \frac{ \langle I_i \rangle }{(1-P_1)^2}\sum\limits_{j}\gamma_j\langle I_j S \rangle \right)\nonumber
\end{eqnarray}
The numerator in the first term on the right-hand side is the rate of change that an individual is infected. Taking the sum over all nodes with the same degree, this gives the rate of change in the expected number of infected individuals with that degree, which is given by Equation~\eqref{eqn:PairPopStandard}. Taking the sum of $\langle I_i \rangle$ over all nodes with the same degree gives the expected number of infected nodes with that degree. Therefore, assuming 
\begin{equation}
\langle I_i \rangle \approx \frac{[I_{k_i}]}{|C_{k_i}|} \quad , \quad T_{ij}=\bar{T}_{k_i k_j} \quad, \quad \gamma_i= \gamma_{k_i} \qquad (i \in \mathcal{V} , j \in \mathcal{N}_i),
\end{equation}
where $[A_k]$ is the expected number of individuals with degree $k$ in state $A$ and $\bar{T}_{kl}$ is the rate of transmission from a degree $l$ to a degree $k$ node, we obtain
\begin{align}
\sum_{i: k_i=k}\sum_{\alpha:\sigma_{\alpha i}=I}\frac{\mbox{d} \rho_\alpha}{\mbox{d} t}=&\frac{\sum_{l \in \mathcal{M}}\bar{T}_{kl}[S_kI_l]-\gamma[I_k]}{1-P_1}+\frac{[I_k]}{(1-P_1)^2}\sum_j\gamma\langle I_jS\rangle,
\label{PopIndiv1}
\end{align}
where $[A_kB_l]$ is the expected number of pairs between individuals of degree $k$ and degree $l$, in states $A$ and $B$ respectively, and $k_i$ is the degree of node $i$. Above, and throughout, all expected numbers are with respect to the standard probability measure $P$. Assuming that the states of individuals are independent, \eqref{PopIndiv1} becomes
\begin{eqnarray}
\sum_{i: k_i=k}\sum_{\alpha:\sigma_{\alpha i}=I}\frac{\mbox{d} \rho_\alpha}{\mbox{d} t}&\approx&\frac{\sum_{l \in \mathcal{M}}\bar{T}_{kl}|C_{k,l}|\frac{[S_k]}{|C_k|}\frac{[I_l]}{|C_l|}-\gamma[I_k]}{1- \prod_{j} \langle S_j \rangle}+\frac{[I_k]}{(1- \prod_{j} \langle S_j \rangle)^2}\sum\limits_{j}\gamma\langle I_j \rangle \prod_{k \neq j} \langle S_k \rangle\nonumber
\label{PopIndiv2}
\end{eqnarray}
where $|C_k|$ is the number of degree $k$ nodes in the network and $|C_{k,l}|$ is the number of pairs between degree $k$ and degree $l$ nodes. This equation is not closed, since the final term and the denominators depend on node-level quantities. However, from \eqref{eqn:PopIndivRelation} the node-level quantities can be approximated by assuming $\langle S_j \rangle = [S_k]/|C_k|$, where $k$ is the degree of node $j$. Therefore
\begin{equation}
\prod_{i} \langle S_i\rangle \approx \prod\limits_{l\in\mathcal{M}}\left(\frac{[S_l]}{|C_l|}\right)^{|C_l|},
\end{equation}
and
\begin{equation}
\gamma\langle I_j \rangle \prod\limits_{i \neq j}\langle S_i \rangle \approx \frac{[I_k]}{|C_k|}\left(\frac{[S_k]}{|C_k|}\right)^{|C_k|-1}\prod\limits_{l \in \mathcal{M}:l \neq k}\left(\frac{[S_l]}{|C_l|}\right)^{|C_l|},
\label{eqn:OneInf}
\end{equation} 
where $k$ is the degree of node $j$. Multiplying Equation~\eqref{eqn:OneInf} by the number of degree $k$ nodes, $|C_k|$, we obtain the probability of a single degree $k$ node being infected, which we denote $\tilde{P}(I_k=1)$. Therefore, we obtain
\begin{eqnarray}
\sum_{i: k_i=k}\sum_{\alpha:\sigma_{\alpha i}=I}\frac{\mbox{d} \rho_\alpha}{\mbox{d} t}&\approx&\frac{\sum_{l \in \mathcal{M}}\bar{T}_{kl}|C_{k,l}|\frac{[S_k]}{|C_k|}\frac{[I_l]}{|C_l|}-\gamma[I_k]}{(1-\prod_{l} (\frac{[S_l]}{|C_l|})^{|C_l|})}+\frac{[I_k]}{(1-\prod_{l} (\frac{[S_l]}{|C_l|})^{|C_l|})^2}\sum\limits_{l\in\mathcal{M}}\gamma \tilde{P}(I_l=1).\nonumber
\label{PopIndiv3}
\end{eqnarray}
To find a steady state, we need to find vectors $\langle X \rangle^*$ and $\langle Y \rangle^*$ satisfying
\begin{eqnarray}
0&=&\frac{\sum_{l \in \mathcal{M}}\bar{T}_{kl}|C_{k,l}|\frac{[X_k]^*}{|C_k|}\frac{[Y_l]^*}{|C_l|}-\gamma[Y_k]^*}{(1-\prod_{l} (\frac{[X_l]^*}{|C_l|})^{|C_l|})}+\frac{[Y_k]^*}{(1-\prod_{l} (\frac{[X_l]^*}{|C_l|})^{|C_l|})^2}\sum\limits_{l\in\mathcal{M}}\gamma \tilde{P}(Y_l =1)^*\nonumber
\label{eqn:PopIndiv4}
\end{eqnarray}
from which we can approximate the expected number of infected degree $k$ individuals in the QSD by computing $ [Y_k]^*/(1-\prod_{l} (\frac{[X_l]^*}{|C_l|})^{|C_l|})$. We require $[Y_k]^* \in [0,|C_k|], [X_k]^*=|C_k|-[Y_k]^*$ for all $i$. Such a solution can be found by defining

\begin{align}
\dot{[Y_k]}=&\sum_{l \in \mathcal{M}}\bar{T}_{kl}|C_{k,l}|\frac{[X_k]}{|C_k|}\frac{[Y_l]}{|C_l|}-\gamma[Y_k]+\frac{[Y_k]\sum\limits_{l\in\mathcal{M}}\gamma \tilde{P}(Y_l=1)}{(1-\prod_{l} (\frac{[X_l]}{|C_l|})^{|C_l|})}\nonumber\\
\ [X_k]=&|C_k|-[Y_k] \nonumber \\
\tilde{P}(Y_k=1)=& |C_k|\frac{[Y_k]}{|C_k|}\left(\frac{[X_k]}{|C_k|}\right)^{|C_k|-1}\prod\limits_{l \in \mathcal{M}:l \neq k}\left(\frac{[X_l]}{|C_l|}\right)^{|C_l|},
\label{eqn:PopIndiv5}
\end{align}
and specifying that $[Y_k(0)]\in [0,|C_k|]$ for all $k$ and calculating the steady-state. Any solution will be a valid solution, since Equation~\eqref{eqn:PopIndiv5} is bounded such that $[Y_k]^* \in [0,|C_k|]$ for all $k$ (this can be shown using a method similar to Appendix~\ref{App:invariant}).

\section{Node-level pair-based QSD model}
\label{App:pair}
\renewcommand{\thefigure}{E\arabic{figure}}
\setcounter{figure}{0} 
\renewcommand{\theequation}{E\arabic{equation}}
\setcounter{equation}{0} 
If we do not assume independence at the level of individuals, we need to find equations describing pair probabilities in the conditional distribution. We have
\begin{eqnarray}
 \frac{\mathrm{d}}{\mathrm{d}t}\left(\rho(\Sigma_i(t)=I)\right)= \sum_{\alpha:\sigma_{\alpha i}=I}\frac{\mbox{d} \rho_\alpha}{\mbox{d} t}&=&\frac{\sum_{j}{T_{ij}}\langle S_iI_j \rangle -\gamma_i \langle I_i \rangle}{1-P_1}+ \frac{\langle I_i\rangle}{(1-P_1)^2}\sum\limits_{j}\gamma_j\langle I_jS \rangle,\nonumber\\
\frac{\mathrm{d}}{\mathrm{d}t}\left(\rho(\Sigma_i(t)=S,\Sigma_j=I)\right)= \sum_{\substack{\alpha:\sigma_{\alpha i}=S, \\ \sigma_{\alpha j}=I}}\frac{\mbox{d} \rho_\alpha}{\mbox{d} t}&=&\frac{\sum_{k \in \mathcal{N}_j \setminus i}{T_{jk}}\langle S_iS_jI_k\rangle}{1-P_1}-\frac{\sum_{k \in \mathcal{N}_i \setminus j}{T_{ik}}\langle I_kS_iI_j\rangle}{1-P_1} \nonumber\\ && \quad \quad -\frac{(T_{ij}+\gamma_j )\langle S_iI_j \rangle}{1-P_1} + \frac{\gamma_i \langle I_iI_j\rangle}{1-P_1}+\frac{\langle S_iI_j\rangle}{(1-P_1)^2}\sum\limits_{j}\gamma_j\langle I_jS\rangle ,\nonumber\\
%\frac{d\left(\rho(\Sigma_i(t)=I,\Sigma_j=I)\right)}{dt}= \sum_{\substack{\sigma:\sigma_i=I,\\ \sigma_j=I}}\frac{\mbox{d} \rho_\sigma}{\mbox{d} t}&=&\frac{\sum_{k \in \mathcal{N}_j \setminus i}{T_{jk}}\langle I_iS_jI_k \rangle}{1-P_1} +\frac{\sum_{k \in \mathcal{N}_i \setminus j}{T_{ik}} \langle I_kS_iI_j \rangle}{1-P_1} -\frac{(\gamma_i +\gamma_j )\langle I_iI_j\rangle}{1-P_1}\nonumber \\
%&& \quad + \frac{T_{ij}\langle S_iI_j \rangle }{1-P_1}+ 
%\frac{T_{ji} \langle I_iS_j \rangle}{1-P_1} + \frac{\langle I_iI_j\rangle}{(1-P_1)^2}\sum\limits_{j}\gamma_j\langle I_jS \rangle. 
\label{eqn:Pair1}
\end{eqnarray}
where $\langle A_i \rangle$ is shorthand for the marginal probability $P(\Sigma_i(t)=A)$ with $A \in \{S,I\}$, $\langle A_i B_j \rangle$ is shorthand for $P(\Sigma_i(t)=A, \Sigma_j(t)=B)$ with $A,B \in \{S,I\}$, $\langle A_i B_j C_k\rangle$ is shorthand for $P(\Sigma_i(t)=A, \Sigma_j(t)=B,\Sigma_k(t)=C)$ with $A,B,C \in \{S,I\}$, and $\langle I_j S \rangle$ is shorthand for $P(\Sigma_j=I,\Sigma_k=S\mbox{ for all } k \neq j)$. We can simplify this system by assuming statistical independence at the level of pairs.

As described in Appendix~\ref{sec:Standard}, we approximate the triples in terms of pairs and individuals by assuming
\begin{equation}{\langle A_iB_jC_k \rangle} \approx \frac{\langle A_iB_j\rangle \langle B_jC_k \rangle}{\langle B_j \rangle}.\nonumber\end{equation}
Under this assumption, Equation~\eqref{eqn:Pair1} becomes
\begin{eqnarray}
\frac{\mathrm{d}}{\mathrm{d}t}\left(\rho(\Sigma_i(t)=I)\right)&=&\frac{\sum_{j}{T_{ij}}\langle S_iI_j \rangle -\gamma_i \langle I_i \rangle}{1-P_1}+ \frac{\langle I_i\rangle}{(1-P_1)^2}\sum\limits_{j}\gamma_j\langle I_jS \rangle,\nonumber\\
\frac{\mathrm{d}}{\mathrm{d}t}\left(\rho(\Sigma_i(t)=S,\Sigma_j=I)\right)&=&\frac{\sum_{k \in \mathcal{N}_j \setminus i}{T_{jk}}\frac{\langle S_iS_j\rangle\langle S_jI_k\rangle}{\langle S_j \rangle}}{1-P_1}-\frac{\sum_{k \in \mathcal{N}_i \setminus j}{T_{ik}}\frac{\langle I_kS_i\rangle\langle S_iI_j\rangle}{\langle S_i \rangle}}{1-P_1}\nonumber \\ && \quad \quad -\frac{(T_{ij}+\gamma_j )\langle S_iI_j \rangle}{1-P_1} + \frac{\gamma_i \langle I_iI_j\rangle}{1-P_1}+\frac{\langle S_iI_j\rangle}{(1-P_1)^2}\sum\limits_{j}\gamma_j\langle I_jS \rangle ,\nonumber\\
%\frac{d\left(\rho(\Sigma_i(t)=I,\Sigma_j=I)\right)}{dt}&=&\frac{\sum_{k \in \mathcal{N}_j \setminus i}{T_{jk}} \frac{\langle I_i S_j \rangle \langle S_j I_k \rangle}{\langle S_j \rangle}}{1-P_1} +\frac{\sum_{k \in \mathcal{N}_i \setminus j}{T_{ik}} \frac{ \langle I_k S_i \rangle \langle S_i I_j \rangle}{ \langle S_i \rangle}}{1-P_1} -\frac{(\gamma_i +\gamma_j )\langle I_iI_j\rangle}{1-P_1} \nonumber\\
%&& \quad + \frac{T_{ij}\langle S_iI_j \rangle }{1-P_1}+ 
%\frac{T_{ji} \langle I_iS_j \rangle}{1-P_1} + \frac{\langle I_iI_j\rangle}{(1-P_1)^2}\sum\limits_{j}\gamma_j\langle I_jS\rangle. 
\label{Pair2}
\end{eqnarray}
Note that $\langle S_i\rangle = 1-\langle I_i \rangle$, $\langle I_iI_j\rangle = \langle I_j\rangle - \langle S_iI_j\rangle$ and $\langle S_iS_j\rangle = \langle S_i\rangle -\langle S_iI_j\rangle$. Both $\langle I_jS \rangle$ and the ground state probability, $P_1$, are full system size, and therefore, following~\cite{Frasca2016,Sharkey2015}, a natural pair approximation for these are
\[\langle I_jS\rangle \approx {\langle \widetilde{I_jS}\rangle}=\frac{\prod\limits_{x \in \mathcal{N}_j}\langle I_jS_x\rangle\prod\limits_{y \neq j}\prod\limits_{x \in \mathcal{N}_y:x<y,x\neq j}\langle S_yS_x \rangle }{\prod\limits_{x \neq j}\langle S_x\rangle^{k_x-1}\langle Y_j \rangle^{k_j-1}}\]
and

\begin{equation}
 P_1  \approx \langle \sigma_1 \rangle =\prod_{y}\prod\limits_{x \in \mathcal{N}_y:x<y} \frac{ \langle S_yS_x \rangle }{ \langle S_y \rangle^{n_y-1}}.\nonumber
\label{eqn:PairGround}
\end{equation}
In the QSD, both the pair level and individual level conditional probabilities are in a steady-state, so both equations in Equation~\eqref{eqn:Pair1} are equal to zero. Therefore, to find the approximation to the QSD under the pair level independence assumption, we need to find vectors $\langle X^*\rangle$, $\langle Y^* \rangle$, and matrices $\langle XX^* \rangle$,$\langle XY^* \rangle$, and $\langle YY^* \rangle$ satisfying,
\begin{eqnarray}
0&=&\frac{\sum_{j}{T_{ij}}\langle X_iY_j \rangle^* -\gamma_i \langle Y_i \rangle^*}{1- \langle \sigma_1 \rangle}+ \frac{\langle Y_i\rangle^*}{(1- \langle \sigma_1 \rangle)^2}\sum\limits_{j}\gamma_j\langle \widetilde{Y_jX} \rangle^*,\nonumber\\
 0&=&\frac{\sum_{k \in \mathcal{N}_j \setminus i}{T_{jk}}\frac{\langle X_iX_j\rangle^*\langle X_jY_k\rangle^*}{\langle X_j \rangle^*}}{1- \langle \sigma_1 \rangle}-\frac{\sum_{k \in \mathcal{N}_i \setminus j}{T_{ik}}\frac{\langle Y_kX_i\rangle^*\langle X_iY_j\rangle^*}{\langle X_i \rangle^*}}{1- \langle \sigma_1 \rangle} \nonumber\\ && \quad \quad -\frac{(T_{ij}+\gamma_j )\langle X_iY_j \rangle^*}{1- \langle \sigma_1 \rangle} + \frac{\gamma_i \langle Y_iY_j\rangle^*}{1- \langle \sigma_1 \rangle}+\frac{\langle X_iY_j\rangle^*}{(1- \langle \sigma_1 \rangle)^2}\sum\limits_{j}\gamma_j\langle \widetilde{Y_jX}\rangle^* ,\nonumber\\
% 0&=&\frac{\sum_{k \in \mathcal{N}_j \setminus i}{T_{jk}} \frac{\langle Y_i X_j \rangle^* \langle X_j Y_k \rangle^*}{\langle X_j \rangle^*}}{1-P_1} +\frac{\sum_{k \in \mathcal{N}_i \setminus j}{T_{ik}} \frac{ \langle Y_k X_i \rangle^* \langle X_i Y_j \rangle^*}{ \langle X_i \rangle^*}}{1-P_1} -\frac{(\gamma_i +\gamma_j )\langle Y_iY_j\rangle^*}{1-P_1}\nonumber \\
%&& \quad + \frac{T_{ij}\langle X_iY_j \rangle^* }{1-P_1}+ 
%\frac{T_{ji} \langle Y_iX_j \rangle^*}{1-P_1} + \frac{\langle Y_iY_j\rangle^*}{(1-P_1)^2}\sum\limits_{j}\gamma_j\langle Y_jX \rangle^*. 
\label{Pair3}
\end{eqnarray}
which, once solved, can be used to find the probability that $i$ is infected in the QSD by computing $\langle Y_i \rangle^*/(1-\langle \sigma_1\rangle^*)$. However, we require solutions $\langle Y_i \rangle^*$ and $\langle X_iY_j \rangle^* \in [0,1]$ which satisfy $\langle X_i \rangle^*=1-\langle Y_i \rangle^*$ for all $i$ and $\langle X_iX_j \rangle = \langle X_i \rangle-\langle X_iY_j \rangle$, and $\langle Y_iY_j\rangle=\langle Y_j\rangle- \langle X_iY_j \rangle$ for all $i,j$ in order to be valid solutions to our original problem. 

By calculating the steady-state of the system,
\begin{eqnarray}
\langle\dot{Y_i}\rangle&=&\sum_{j}{T_{ij}}\langle X_iY_j \rangle-\gamma_i \langle Y_i\rangle+ \frac{\langle Y_i\rangle\sum\limits_{j}\gamma_j \langle \widetilde{Y_jX}\rangle}{1- \langle \sigma_1 \rangle},\nonumber\\
\langle\dot{X_iY_j}\rangle&=&\sum_{k \in \mathcal{N}_j \setminus i}{T_{jk}}\frac{\langle X_iX_j\rangle\langle X_jY_k\rangle}{\langle X_j \rangle}-\sum_{k \in \mathcal{N}_i \setminus j}{T_{ik}}\frac{\langle Y_kX_i\rangle\langle X_iY_j\rangle}{\langle X_i \rangle}\nonumber\\
&&-(T_{ij}+\gamma_j )\langle X_iY_j \rangle + \gamma_i \langle Y_iY_j\rangle+\frac{\langle X_iY_j\rangle\sum\limits_{j}\gamma_j\langle \widetilde{Y_jX} \rangle}{1- \langle \sigma_1 \rangle} ,\nonumber\\
%%\dot{\langle Y_i^cY_j^c \rangle}&=&\Big[1-\langle\sigma_0\rangle\Big]\Big[\sum_{k \in \mathcal{N}_j \setminus i}{T_{jk}} \frac{\langle Y_i^c X_j^c \rangle \langle X_j^c Y_k^c \rangle}{\langle X_j^c \rangle} +\sum_{k \in \mathcal{N}_i \setminus j}{T_{ik}} \frac{ \langle Y_k^c X_i^c \rangle \langle X_i^c Y_j^c \rangle}{ \langle X_i^c \rangle}\nonumber \\
%%&&-(\gamma_i +\gamma_j )\langle Y_i^cY_j^c\rangle+ T_{ij}\langle X_i^cY_j^c \rangle + 
%%T_{ji} \langle Y_i^cX_j^c \rangle\Big]+ \langle Y_i^cY_j^c\rangle\frac{\sum\limits_{j}\gamma_j\langle Y_jX \rangle}{1-\langle \sigma_0^c\rangle}, \nonumber\\ 
\langle X_i \rangle&=&1- \langle Y_i \rangle,\nonumber \\
\langle X_iX_j \rangle &=& \langle X_i \rangle- \langle X_iY_j \rangle, \nonumber \\
\langle Y_iY_j \rangle &=& \langle Y_i \rangle-\langle Y_iX_j \rangle, 
\label{eqn:PairNode5}
\end{eqnarray}
where
\[\langle \widetilde{Y_jX}\rangle = \frac{\prod\limits_{x \in \mathcal{N}_j}\langle Y_jX_x\rangle\prod\limits_{y \neq j}\prod\limits_{x \in \mathcal{N}_y:x<y,x\neq j}\langle X_yX_x \rangle }{\prod\limits_{x \neq j}\langle X_x\rangle^{k_x-1}\langle Y_j \rangle^{k_j-1}}\]
and

\begin{equation}
\langle \sigma_1 \rangle = \prod_{y}\prod\limits_{x \in \mathcal{N}_y:x<y} \frac{ \langle X_yX_x \rangle }{ \langle X_y \rangle^{n_y-1}}.\nonumber
\label{eqn:PairGround}
\end{equation}
we can approximate the probability that $i$ is infected in the QSD by computing $\lim_{t \to \infty} \langle Y_i(t) \rangle^*/(1-\langle \sigma_0(t) \rangle^*)$.

\section{Population-level pair-based QSD model}
\label{App:population}
\renewcommand{\thefigure}{F\arabic{figure}}
\setcounter{figure}{0} 
\renewcommand{\theequation}{F\arabic{equation}}
\setcounter{equation}{0} 
To obtain a population-level pair-based model, we sum over nodes with the same degree (and pairs of nodes with same pair of degrees); i.e.
\begin{eqnarray}
\sum_{i: k_i=k}\sum_{\alpha:\sigma_{\alpha i}=I}\frac{\mbox{d} \rho_\alpha}{\mbox{d} t}&=&\frac{\tau\sum_{l \in \mathcal{M}}[S_kI_l]-\gamma[I_k]}{1-P_1}+\frac{[I_k]}{(1-P_1)^2}\sum_j\gamma\langle I_jS\rangle \nonumber\\
\sum_{\substack{i,j: k_i=k, \\ k_j=l}}\sum_{\substack{\alpha:\sigma_{\alpha i}=S,\\ \sigma_{\alpha j}=I}}\frac{\mbox{d} \rho_\alpha}{\mbox{d} t}&=&\frac{\tau\sum_{m \in \mathcal{M}}[S_kS_lI_m]-\tau\sum_{m \in \mathcal{M}}[I_mS_kI_l]-\tau[S_kI_l]+\gamma[I_kI_l]-\gamma[S_kI_l]}{1-P_1}\nonumber\\
&&+\frac{[S_kI_l]}{(1-P_1)^2}\sum_j\gamma\langle I_jS\rangle, \nonumber\\
%\sum_{\substack{i,j: k_i=k, \\ k_j=l}}\sum_{\substack{\sigma:\sigma_i=I,\\ \sigma_j=I}}\frac{\mbox{d} \rho_\sigma}{\mbox{d} t}&=&\frac{\tau\sum_{m \in \mathcal{M}}[I_kS_lI_m]+\tau\sum_{m \in \mathcal{M}}[I_mS_kI_l]+\tau[S_kI_l]+\tau[I_kS_l]-2\gamma[I_kI_l]}{1-P_1}\nonumber\\
%&&+\frac{[S_kS_l]}{(1-P_1)^2}\sum_j\gamma\langle I_jS\rangle
\label{eqn:PairPop}
\end{eqnarray}
where $[A_kB_lC_h]$ is the expected number of triples between degree $k$, degree $l$ and degree $h$ individuals in states $A$, $B$ and $C$ respectively.

As described in Appendix~\ref{sec:Standard}, we can express the triple terms as
\begin{equation}
[A_lS_kI_m] \approx \frac{k-1}{k}\frac{[A_lS_k][S_kI_m]}{[S_k]},
\label{eqn:triple1}
\end{equation}
We can set equations~\eqref{eqn:PairPop} to zero and use the approximation~\eqref{eqn:triple1} to find equations describing the QSD. 

A solution to the resulting system can be found by finding an steady-state of
\begin{align}
[\dot{Y_k}]=&-\gamma[Y_k]+\tau\sum_{l\in\mathcal{M}}[X_kY_l] +\frac{[Y_k]\sum\limits_{l\in\mathcal{M}}\gamma \tilde{P}(Y_l=1)}{1- \langle \sigma_1 \rangle}\nonumber\\
{}[\dot{X_kY_l}]=& \tau(\sum_{m\in\mathcal{M}}\frac{l-1}{l}\frac{[X_kX_l][X_lY_m]}{[X_l]}-\sum_{m\in\mathcal{M}}\frac{k-1}{k}\frac{[Y_mX_k][X_kY_l]}{[X_k]}\nonumber\\
&-[X_kY_l])+\gamma([Y_kY_l]-[X_kY_l])+\frac{[X_kY_l]\sum\limits_{l\in\mathcal{M}}\gamma \tilde{P}(Y_l=1)}{1- \langle \sigma_1 \rangle}\nonumber \\
{}[\dot{Y_kY_l}]=&\tau(\sum_{m\in\mathcal{M}}\frac{k-1}{k}\frac{[Y_mX_k][X_kY_l]}{[X_k]}+\sum_{m\in\mathcal{M}}\frac{l-1}{l}\frac{[Y_kX_l][X_lY_m]}{[X_l]})\nonumber\\
&+\tau([X_kY_l]+[Y_kX_l]-2\gamma[Y_kY_l]+\frac{[Y_kY_l]\sum\limits_{l\in\mathcal{M}}\gamma \tilde{P}(Y_l=1)}{1- \langle \sigma_1 \rangle}  \nonumber \\
\ [X_k]=&|C_k|-[Y_k]\nonumber \\
\ [X_kX_l]=&|C_{k,l}|-[Y_kY_l]-[X_kY_l]-[X_lY_k],
\label{eqn:PairPop5}
\end{align}
where $\tilde{P}(Y_l=1)=|C_l|\langle Y_iX \rangle$ for some $i$ with $k_i=l$. Here
\begin{equation}\langle Y_iX \rangle = \frac{\prod\limits_{x,y \neq i}G_{xy}\langle X_xX_y \rangle \prod\limits_{x}G_{ix}\langle Y_iX_x\rangle}{\prod\limits_{x \neq i}\langle X_x\rangle^{k_x-1}\langle Y_i \rangle^{k_i-1}},\end{equation}
which requires node-level terms. We can approximate this by population-level quantities using
\begin{align}
\langle S_i \rangle \approx \frac{[S_{k_i}]}{|C_{k_i}|},
\label{eqn:PopIndivRelation2}
\end{align}
and
\begin{equation}
\langle S_iS_j \rangle \approx \frac{[S_{k_i}S_{k_j}]}{|C_{k_i,k_j}|},
\label{eqn:PopPairRelation2}
\end{equation}
based on the discussion in Appendix~\ref{sec:Standard}. This gives
\begin{equation}
\langle Y_jX \rangle \approx \frac{\prod\limits_{k \neq k_j}\prod\limits_{l \leq k:l\neq k_j}\left(\frac{[X_kX_l]}{|C_{k,l}|}\right)^{|C_{k,l}|}\left(\frac{[Y_{k_j}X_k]}{|C_{k_j,k}|}\right)^{\frac{|C_{k_j,k}|}{|C_{k_j}|}}\left(\frac{[X_{k_j}X_k]}{|C_{k_j,k}|}\right)^{|C_{k_j,k}|-\frac{|C_{k_j,k}|}{|C_{k_j}|}} }{\prod\limits_{k \neq k_j}\left(\frac{[X_k]}{|C_k|}\right)^{|C_k|(k-1)}\left(\frac{[Y_{k_j}]}{|C_{k_j}|}\right)^{(k_j-1)}\left(\frac{[X_{k_j}]}{|C_{k_j}|}\right)^{(|C_{k_j}|-1)(k_j-1)}}.\label{eqn:PopFlow}
\end{equation}
To approximate the ground state recall that in the previous section we have shown that a natural approximation to the ground state probability under the assumption of pair level independence is
\begin{equation}
 \langle \sigma_1 \rangle \approx \prod_{i} \prod_{j<i} \frac{ G_{ij} \langle X_i X_j \rangle }{ \langle X_i \rangle^{n_i-1}}.\nonumber
\end{equation}
Using equations~\eqref{eqn:PopIndivRelation2} and \eqref{eqn:PopPairRelation2} we can approximate this in terms of population level quantities, which yields
\begin{equation}
\langle \sigma_1 \rangle \approx \prod\limits_{k}\prod\limits_{l \leq k} \left(\frac{\left(\frac{[X_kX_l]}{|C_{k,l}|}\right)^{|C_{k,l}|}}{\left(\frac{[X_k]}{|C_k|}\right)^{|C_k|(k-1)}}\right)\label{eqn:PopGround}
\end{equation}
By substituting equations~\eqref{eqn:PopGround} and~\eqref{eqn:PopFlow} into Equation~\eqref{eqn:PairPop5} we obtain a closed system of equations.

\end{appendices}

\section*{Acknowledgements}
\noindent CO and KS acknowledge support from EPSRC grant (EP/N014499/1). The authors would like to thank Ian Smith for use of the ARC Condor high throughput computing system at the University of Liverpool \texttt{http://condor.liv.ac.uk/}, which significantly sped up simulation of the stochastic models. 

\section*{Author contributions}
\noindent CO, KS and RW created the project, performed the analysis and wrote the manuscript. JM created the project and perfomed the analysis. AL created the project.

\section*{Competing interests}
\noindent The authors have no competing interests to declare.

\section*{Data and materials}
\noindent Matlab code for solving the models will be published online with the manuscript. Python code will be added to the Epidemics on Networks package.

\begingroup
\let\itshape\upshape
\bibliographystyle{plain}
\bibliography{refs}
\endgroup
\end{document}